\documentclass[a4paper,10pt]{article}
\usepackage{amsfonts}
\usepackage{amsmath}
\usepackage{indentfirst,latexsym,bm}
\usepackage{amssymb}
\usepackage{amsmath}
\usepackage{natbib}
\usepackage{graphics}
\usepackage{authblk}
\usepackage{color}
\usepackage{float}
\usepackage{latexsym}
\usepackage{booktabs}
\usepackage{multirow}
\usepackage{graphicx}
\usepackage{epsfig}
\usepackage{geometry}
\usepackage{amscd}
\usepackage{lscape}
\usepackage{float}
\usepackage{amsmath}
\usepackage{amssymb}
\usepackage{tabularx}
\usepackage{enumerate}
\usepackage{amsfonts}
\usepackage{mathrsfs}
\usepackage{eufrak}
\usepackage{graphicx}
\usepackage{dsfont}
\usepackage{slashbox}
\usepackage{epsfig}
\usepackage[ruled, lined, vlined, shortend, algosection]{algorithm2e}
\usepackage{longtable}
\usepackage{pdflscape}
\usepackage{lscape}
\usepackage{curves}
\usepackage{epic}
\usepackage[table]{xcolor}
\usepackage{booktabs}
\usepackage{subfigure}
\usepackage[all]{xy}
 \renewcommand\appendix{\par
   \setcounter{section}{0}%
   \setcounter{subsection}{0}%
   \setcounter{figure}{0}%
   \renewcommand\thesection{\Alph{section}}%
   \renewcommand\thefigure{\arabic{figure}}}

\newcommand{\pa}{\partial}
\newcommand{\fs}[2]{{\dfrac{#1}{#2}}}
\numberwithin{equation}{section}

\newtheorem{assumption}{Assumption}[section]
\newtheorem{theorem}{Theorem}[section]

\newtheorem{lemma}[theorem]{Lemma}

\newtheorem{proposition}[theorem]{Proposition}

\newenvironment{proof}[1][Proof]{\textbf{#1.} }{\ \rule{0.5em}{0.5em}}
\DeclareMathSizes{12}{11}{7.7}{5.5}
\begin{document}

\title{\textbf{Funding Liquidity, Debt Tenor Structure,\\ and Creditor's Belief: An
Exogenous\\ Dynamic Debt Run Model} }
\date{}

\author[$a$]{\small\textsc{{Gechun Liang\footnote{Corresponding author. Tel.: +44 020 7848 2633}}}}
\affil[$a$]{\small{\textsc{Department of Mathematics, King's College London, Strand, London, WC2R 2LS, U.K.}}\\

\texttt{gechun.liang@kcl.ac.uk}}\vspace{0.2cm}

\author[$b$]{\small\textsc{{Eva L\"utkebohmert}}}
\affil[$b$]{\small\textsc{{Department of Quantitative Finance, University of Freiburg, Platz der Alten Synagoge, 79098 Freiburg, Germany}}\\

\texttt{eva.luetkebohmert@finance.uni-freiburg.de}}\vspace{0.2cm}

\author[$c$]{\small\textsc{Wei Wei}}
\affil[$c$]{\small\textsc{{Mathematical Institute, University of
Oxford, Oxford,
OX2 6GG, U.K.}}\\

\texttt{wei.wei@maths.ox.ac.uk}}\vspace{0.2cm}

\maketitle

\begin{abstract}
We propose a unified structural credit risk model incorporating both
insolvency and illiquidity risks, in order to investigate how a
firm's default probability depends on the liquidity risk associated
with its financing structure. We assume the firm finances its risky
assets by mainly issuing short- and long-term debt. Short-term debt
can have either a discrete or a more realistic staggered tenor
structure. At rollover dates of short-term debt, creditors face a
dynamic coordination problem. We show that a unique threshold
strategy (i.e., a debt run barrier) exists for short-term creditors
to decide when to withdraw their funding, and this strategy is
closely related to the solution of a non-standard optimal stopping
time problem with control constraints. We decompose the total credit
risk into an insolvency component and an illiquidity component based
on such an endogenous debt run barrier together with an exogenous
insolvency barrier.
\end{abstract}

\noindent {\it Key words}\/: Structural credit risk model, debt run,
liquidity risk, first passage time, optimal stopping time\\

\noindent {\it JEL}\/ Codes: G01, G20, G32, G33 \\

\newpage

\section{Introduction}\label{SecIntroduction}

The recent financial crisis has dramatically shown that financial
markets are not ideal. In particular, refinancing in periods of
financial distress can be extremely costly or even impossible due to
liquidity drying up in the market. It has been shown, for example by
\cite{AdrianShin2008,AdrianShin2010} and \cite{Brunnermeier2009},
that the heavy use of short-term debt was a key contributing factor
to the credit crunch of 2007/2008. Firms, however, often prefer
short-term debt financing as it is cheaper than long-term debt.
Moreover, as argued by \cite{HeXiongRun}, short-term debt can also
be regarded as a disciplinary device for firms and can be used to
mitigate adverse selection problems and to reduce the cost of
auditing firms. Hence, several reasons support the use of short-term
financing. However, most of the existing credit risk models do not
take into account the rollover risk (or liquidity risk) inherent in
short-term debt financing. It is the aim of this paper to provide a
unified framework that incorporates rollover risk as well as
insolvency risk. Within an extended structural credit risk model,
our approach allows to investigate how a firm's default probability
depends on
the rollover risk inherent in its particular financing structure.\\

Structural credit risk models were initiated by \cite{Merton} and
\cite{BlackCox}. In these models default happens if the firm
fundamental falls below some exogenous default barrier which often
relates to the firm's debt level. A huge part of the literature on
structural credit risk modeling focuses on how to model such an
exogenous default barrier, as in \cite{LongstaffSchwartz} and
\cite{Briys}, among others\footnote{Optimal capital structural
models are regarded as the second generation of structural credit
risk models, which were initiated by \cite{Leland,Leland2}, and
\cite{LelandToft}. Therein the firm defaults when its equity value
drops to zero, and the default barrier is determined endogenously by
its equity holders. \cite{Rogers} and \cite{Chen} extend this model
by introducing jump risk, and recently, \cite{HeXiongRolloverRisk}
extended this framework by including an illiquid debt market.}. In
the following we will call this exogenous default barrier the
\emph{insolvency barrier}. Given this exogenous insolvency barrier,
in this paper we derive an endogenous threshold value {below} which
short-term creditors decide to withdraw their funding, i.e., to run
on the firm, and we will call this barrier the \emph{debt run
barrier}. The latter depends on {not only} the firm's
creditworthiness but also the creditors' beliefs about the
likelihood of a debt run {in the remaining rollover periods}.
Determining this debt run barrier is the main problem of this paper.
There is a third barrier, called the \emph{illiquidity barrier},
which represents the critical value when the firm is unable to pay
off its creditors in case of a debt run, and which is determined
endogenously from the debt run barrier. In addition, we show that
the debt run barrier always dominates the illiquidity barrier, which
in turn dominates the insolvency barrier. This relationship among
all three barriers not only helps to decompose the total credit risk
into an insolvency component and an illiquidity component, but also
illustrates the phenomenon that most firms have defaulted due to
illiquidity rather than due to insolvency in the recent
credit crunch.\\

Our first contribution is the provision of a rigorous formulation
for a class of structural credit risk models that study debt runs.
The {classic} debt/bank run model of \cite{DiamondDybvig} features a
static setting where all the depositors simultaneously decide
whether or not to withdraw their demand deposits from a solvent but
illiquid bank. \cite{EricssonRenault} and \cite{Goldstein} provide
further extensions that are, however, still in the static setting.
In this paper, we consider debt runs from a dynamic viewpoint.  The
debt run model introduced by \cite{MorrisShin} focuses on a
two-period setting where short-term creditors face a binary decision
in terms of global games at an interim time point.
\cite{LiangLutkebohmertXiao} provide a structural credit risk model
that also takes liquidity risk into account, as short-term creditors
can decide at a finite number of decision dates whether to roll over
or to withdraw their funding. They derive a debt/bank run barrier
based on the comparison of binary strategies for a representative
short-term creditor. Technically, the generalization from the
two-period setting of \cite{MorrisShin} towards the multi-period
setting of \cite{LiangLutkebohmertXiao} relies on the dynamic
programming principle (DPP). In \cite{LiangLutkebohmertXiao} the DPP
was only applied informally by comparing the expected returns for
the two investment options of creditors at the rollover dates. In
this paper, by introducing an appropriate value function for a
representative short-term creditor, which describes the discounted
expected return over the remaining rollover periods and which is
calculated based on the DPP, we derive the unique threshold
strategy, i.e., the debt run barrier. The representative creditor
decides to withdraw her funding if the firm's fundamental falls
below this barrier at any decision date. In contrast to
\cite{LiangLutkebohmertXiao}, the corresponding dynamic programming
equations presented in this paper are more generic and transparent,
which in particular allows us to introduce flexible
debt maturity structures into our model.\\

The second contribution of this paper is the imbedding of flexible
debt tenor structures into such an extended structural credit risk
model. In \cite{LiangLutkebohmertXiao} a discrete tenor structure is
assumed such that the rollover dates of short-term debt are given by
a sequence of deterministic numbers. This implies that all
short-term debt expires and can be rolled over at the same time. The
problem is therefore equivalent to a one-creditor problem. In
reality, however, firms typically stagger the maturities for
short-term debt to finance their long-term risky assets. Rollover
risk is partially reduced in this way as at each maturity date only
a fraction of total debt is due. Nevertheless, due to the maturity
mismatch between the assets and the liabilities sides, the firm is
still exposed to significant liquidity risk. Our model covers both
the discrete tenor structure and the staggered tenor structure. The
latter was first introduced by \cite{Leland,Leland2}, and
\cite{LelandToft}, with \cite{Rogers} and \cite{Chen} providing
further technical details. The main idea is to assume a random
duration of debt in order to reflect the maturity mismatch.
Recently, \cite{HeXiongRun} applied the staggered maturity structure
to the debt run literature where debt maturities are modelled as
arrival times of a Poisson process, whose intensity can then be
interpreted as the inverse of average debt duration. In this paper,
we utilize a more general and flexible Cox process to model the
staggered tenor structure. The economic intuition of using the Cox
maturity structure is that the average duration of the short-term
debt which a firm issues should fluctuate and depend on some
economic factors such as the firm fundamental or even the underlying
systemic risk from the market.\\

The third contribution of this paper is that we use a reduced-form
approach to model the impacts of other creditors' rollover decisions
on the representative short-term creditor's rollover decision. In
general, such impacts and the resulting equilibria are complicated
as we need to consider not only the impact of other creditors'
current rollover decisions, but also the impact of their future
decisions. In \cite{MorrisShin}, they concentrate on the former
impact, and resort to a reduced-form approach by modeling the
representative short-term creditor's belief on the proportion of
creditors not rolling over their funding at each rollover date as a
uniformly distributed random variable exogenously. In
\cite{HeXiongRun}, they concentrate on the latter impact, and assume
at each small time interval, there is only a small proportion of
creditors whose contracts expire and need to be rolled over. In this
paper, we try to include both impacts. The former impact is modeled
in a similar way to \cite{MorrisShin} and
\cite{LiangLutkebohmertXiao} by assuming the representative
creditor's belief exogenously. However, we take a more general model
in the sense that the representative creditor's belief is a general
random variable, and we also compare the results with different
assumptions on the distribution of such a random variable. The
latter impact is included by considering the dynamic programming
equation of the representative creditor, which is in the same spirit
as \cite{HeXiongRun}.\\

Finally, our fourth contribution is to answer the question whether
the representative short-term creditor's rollover decision is indeed
optimal. This question also arises in the existing dynamic debt run
models such as those in \cite{MorrisShin}, \cite{HeXiongRun}, and
\cite{LiangLutkebohmertXiao}, but has so far not been answered. In
both the discrete and the staggered tenor structures, we show that
the decision problem of the representative short-term creditor is
equivalent to a non-standard optimal stopping time problem with
control constraints. At each rollover date the representative
creditor faces the risk that the firm may fail due to a debt run
based on her belief. If the firm survives, the creditor can then
decide whether to withdraw her funding (stop) or to roll over her
contract (continue). If the firm fails due to other creditors' runs,
the representative creditor is then forced to stop and faces the
recovery risk from bankruptcy. Therefore, the decision time for the
representative creditor must exclude the default time due to debt
runs. For the case of the staggered tenor structure, since the
maturity dates are the arrival times of a Cox process, the
representative creditor is only allowed to stop (i.e., to withdraw
her funding) at a sequence of Cox arrival times rather than at any
stopping time.  In the literature, such kind of optimal stopping at
Poisson-type arrival times has been used to solve the standard
optimal stopping time problem by
\cite{Krylov} as the so-called randomized stopping time technique.\\


The paper is organized as follows. Section \ref{SecBenchmarkModel}
describes the assumptions on the firm's capital structure and
explains the rollover decision of a representative short-term
creditor in the
benchmark model. 
We also present the rigorous formulation of the rollover decision
problem in terms of dynamic programming equations. In section
\ref{sec barriers} we use the creditor's value function derived in
the dynamic programming equations to determine the short-term
creditor's debt run barrier as well as the firm's illiquidity
barrier in case of both the discrete and the staggered debt
structures. We reformulate the creditor's decision problem in terms
of the associated optimal stochastic control problem in section
\ref{sec optimal control}. Section \ref{SecConclusion} discusses the
related literature and concludes.

\section{Benchmark Debt Run Model}\label{SecBenchmarkModel}

In this section, we propose a debt run model that incorporates
rollover risk into the structural credit risk framework.

\subsection{Capital Structure of a Firm}
Consider a market defined over a complete probability space
$(\Omega,\mathcal{F},\mathbf{P})$, which supports a standard
Brownian motion $(W_t)_{t\geq 0}$ with its natural filtration
$\{\mathcal{F}_t\}$ after augmentation. The market interest rate $r$
is assumed to be constant. In this market, consider a firm whose
fundamental value of assets follows
\begin{equation*}
\frac{dV_t}{V_t}=r_Vdt+\sigma dW_t,
\end{equation*}
with constant volatility $\sigma>0$. The constant $r_V$ denotes the
expected return on the firm's risky assets. We assume that the firm
fundamental is publicly observable.

The firm finances its asset holdings in the duration $[0,T]$ by
issuing short-term debt, such as asset-backed commercial papers and
overnight repos, long-term debt such as corporate bonds, and
equities and others. At initiation time $T_0=0$ an amount $L_0$ is
borrowed long-term at rate $r_L$ until fixed maturity $T>0$.
Moreover, an amount $S_0$ is borrowed short-term at rate $r_S$ until
maturity $T_1$. When short-term debt matures it can be successively
rolled over until the next rollover date. This produces a sequence
of maturity dates (or rollover dates) $0=T_0<T_1<T_2
\cdots<T_{\infty}=\infty$ for short-term debt. For the moment, we do
not impose any structural conditions on the short-term debt
maturities $\{T_{n}\}_{n\geq 1}$. They could be either deterministic
or random. 
%

If there is no default, the value of short-term debt follows
\begin{equation*}
dS_{t}=r_{S}S_{t}dt,
\end{equation*}%
and the value of long-term debt follows
\begin{equation*}
dL_{t}=r_{L}L_{t}dt.
\end{equation*}
The ratio of long-term debt over short-term debt $L_t/S_t$ is
denoted by $l_t$ and follows
\begin{equation*}
dl_{t}=(r_{L}-r_S)l_{t}dt.
\end{equation*}
Moreover, we introduce a process $X_t$ as the ratio of the firm's
asset value over the short-term debt value $X_t=V_t/S_t$. Hence,
$X_t$ follows
$$\frac{dX_t}{X_t}=(r_V-r_S)dt+\sigma dW_t.$$

Short-term creditors have the opportunity to withdraw their funding
at the rollover dates. When the firm is under financial distress or
when an outside investment opportunity is more attractive they will
make use of this option. Long-term creditors, however, are locked in
once they lend money to the firm. They are exposed to a higher risk,
and therefore, should be rewarded with a higher interest rate.
Moreover, since creditors are exposed to the firm's default risk, a
risk premium should be paid on top of the market interest rate. We
have the following assumption on different interest
rates\footnote{The assumption of constant interest rates is imposed
to simplify derivations. In reality, different rates not only vary
in time, but also move differently, motivating the so called
multi-curve modeling (see for example \cite{Crepy}).}.

\begin{assumption}\label{Assumption1}
The long-term interest rate $r_L$ is strictly greater than the
short-term interest rate $r_S$, while the latter is strictly greater
than the market interest rate $r$, i.e., we assume $r_L>r_S>r$.
\end{assumption}

\subsection{The Rollover Decision of a Representative Short-Term Creditor}

Short-term creditors choose whether to renew their maturing
contracts, that is, they need to decide whether to roll over or to
withdraw their funding (i.e., to run) at the maturity times. Hence,
they face a dynamic coordination problem.

Consider the decision problem of a \emph{representative short-term
creditor}. The first key factor to determine the representative
short-term creditor's rollover decision is the insolvency risk
stemming from the deterioration of the firm fundamental. To include
this factor, we follow the classic first-passage-time framework (see
for example \cite{BlackCox}) by assuming an exogenously given
\emph{insolvency barrier}
\begin{equation*}
D_{t}^{Ins}=S_t\beta (l_t),
\end{equation*}%
where $\beta:(0,\infty)\rightarrow(0,\infty)$ is a \emph{safety
covenant function} of the ratio $l_t=L_t/S_t$. As long as the asset
value $V_t$ at any time $t$ is greater than or equal to the total
value of debt $S_t+L_t$, the firm can be considered solvent. Hence,
it is natural to assume that
\begin{equation}\label{equ ass beta}
\beta(l_t)\leq (1+l_t)
\end{equation}
such that
$$
D_t^{Ins}=S_t\beta(l_t)\leq S_t (1+l_t)=S_t+L_t.
$$
The bankruptcy time due to insolvency is then given by the following
first-passage-time
$$\tau^{Ins}=\inf\{t\geq 0:V_t\leq D_t^{Ins}\}
=\inf\{t\geq 0:X_t\leq \beta(l_t)\}.$$

{To coordinate with other creditors, the representative creditor
needs to take other creditors' rollover decisions into account, and
makes her own decision based on whether the firm will survive debt
runs or not at each rollover date $T_n$. Assume that the creditor
believes that the proportion of short-term creditors not rolling
over their funding at each rollover date $T_n$ is a random variable
$\xi$ supported on $[0,1]$ with its conditional density
$f(\cdot|X_{T_n})$ given $X_{T_n}$.}

{The firm will survive debt runs if it can raise enough funding to
pay off its creditors who run on the firm and still keep solvent. In
case of a debt run the firm has to issue collateralized debt by
pledging its assets as collateral to raise the liquidity. The actual
value of the collateral does not matter. It is the maximum value of
the collateral that determines whether the firm is still liquid or
not. The maximum collateral value of the assets is expressed in
terms of the fire-sale price $\psi V_{T_n}$ with the fire-sale rate
$\psi\in(0,1)$\footnote{The constant $\psi$ is the fire-sale rate of
the firm fundamental when the firm is in a distressed state, i.e.,
it represents the amount that can be borrowed by pledging one unit
of the risky assets as collateral. For a detailed discussion of how
to endogenously determine the fire-sale rate by the leverage of the
firm, we refer to \cite{LiangLutkebohmertXiao}. }. If $\psi
V_{T_n}\geq \xi S_{T_n}$, the firm is able to pay off its creditors
who run on the firm, so a potential debt run at time $T_n$ would not
lead to a default. Hence, the second key factor determining the
creditor's rollover decision is her belief about the
\emph{probability that the firm survives the debt run at the
rollover date $T_n$}, which equals
\begin{align}\label{survival_prob}
\theta(X_{T_n})
&=\mathbf{P}\left(\psi V_{T_n}\geq \xi S_{T_n}|X_{T_n}\right)\nonumber\\
&=\mathbf{P}\left(\psi X_{T_n}\geq \xi|X_{T_n}\right)\nonumber\\
&=\mathbf{P}\left(0\leq \xi\leq \min\{1,\psi
X_{T_n}\}|X_{T_n}\right)=\int_0^{\min\{1,\psi
X_{T_n}\}}f(x|X_{T_n})dx
\end{align}
conditional on the firm being solvent at $T_n$, i.e. on the event
$\{V_{T_n}\geq D^{Ins}_{T_n}\}$. In both \cite{MorrisShin} and
\cite{LiangLutkebohmertXiao}, the random variable $\xi$ is simply
assumed to be uniformly distributed on $[0,1]$, so that
$f(\cdot|X_{T_n})=1$ and
$$\theta(X_{T_n})=\min\{1,\psi
X_{T_n}\}.$$ This assumption is justified by global games theory as
it has been shown in \cite{MorrisShin2003} that this is a limiting
case of the situation with unobservable firm fundamental when the
variance of the noise term tends to zero. In section
\ref{SubSecMaturity}, we will show that the uniform distribution
assumption is relatively robust by comparing
it to a family of truncated normal distributions with the same mean but different variances.}\\

The third key factor for the representative short-term creditor's
rollover decision is the recovery rate when the firm defaults either
due to debt runs or due to insolvency. If the firm defaults at some
time $t\in[0,T]$, the firm is exposed to certain bankruptcy costs.
Suppose these are proportional to the firm fundamental value, and
for $\alpha\in(0,1)$, $\alpha V_t$ is the firm value after having
paid the bankruptcy costs. Then, the value $\alpha V_t$ will be
divided among all the creditors, so the representative short-term
creditor obtains the proportion of her funding and she gets at most
her debt value back. Thus, we define the recovery rate as
\begin{equation}\label{equ recovery rate}
R_{t}=\min\left\{1,\frac{\alpha V_{t}}{S_{t}+L_{t}}\right\}=
\min\left\{1,\frac{ \alpha X_{t}}{1+l_{t}}\right\}.
\end{equation}

Note that in case of a default due to a debt run, it can happen that
the asset value is larger than $S_t+L_t$. Therefore, we have to cut
off the recovery rate by $1.$ However, if the firm defaults due to
insolvency at the first-passage-time $\tau^{Ins}$, the asset value
by definition equals the insolvency barrier $D^{Ins}_{\tau^{Ins}}$.
In this case the recovery rate equals
$$R_{\tau^{Ins}}=\frac{\alpha D_{\tau^{Ins}}^{Ins}}{S_{\tau^{Ins}}+L_{\tau^{Ins}}}
=\frac{ \alpha \beta(l_{\tau^{Ins}})}{1+l_{\tau^{Ins}}}$$ which is
less than 1 by condition (\ref{equ ass beta}).


%

We assume that rollover decisions are solely determined by the
aforementioned factors, which we summarize as follows.
\begin{assumption}\label{Assumption2}
The following three factors determine the rollover decision of a
representative short-term creditor.
\begin{itemize}
\item[(i)] Insolvency risk is reflected by the first-passage-time $\tau^{Ins}$
when the firm's asset value falls below the insolvency barrier
$D_t^{Ins}$.
\item[(ii)] Rollover risk is reflected by the
representative short-term creditor's belief on the probability
$\theta(X_{T_{n}})$ that the firm survives the debt run at the
rollover date $T_n$.
\item[(iii)] Recovery risk is reflected by the fraction $R_t$ of
funding that the representative short-term creditor obtains in case
of a default at time $t$.
\end{itemize}
\end{assumption}

We further impose the following condition on the safety covenant
function for technical convenience.

\begin{assumption}\label{Assumption3}
The safety covenant function $\beta(l_t)$ in the definition of the
insolvency barrier $D_t^{Ins}$ has the linear form $\beta(l_t)=\beta
l_t$ for some positive constant $\beta\leq (1/l_t+1)$.
\end{assumption}

\subsection{Dynamic Programming Equations}

In this section we derive dynamic programming equations for the
short-term creditor's rollover decision problem. We consider a
\emph{representative short-term creditor} who invests an amount
normalized to 1 monetary unit at time $t\in[0,T]$. Her discounted
expected return over the remaining time period $[t,T]$ is described
by the value function $U(t,x)$ given the current ratio $X_t=x$ of
asset value over short-term debt value, and she discounts at the
market rate $r$. To investigate the creditor's value function we go
backwards in time starting with her last rollover date prior to
terminal time $T$. Suppose that her $N^{\rm th}$ rollover date is
the closest one prior to the maturity $T$ of long-term debt, that
is, $T_{N}<T$ and $T_{N+1}\geq T$. Figure \ref{fig:maturity}
illustrates the maturities of the short- and long-term debt. If the
maturities are random (as in section \ref{SubSecMaturity}), then
what we do in the following is to condition on one realization of
$\{T_n(\omega)\}_{n\geq 1}$ so that $T_{N}(\omega)<T\leq
T_{N+1}(\omega)$. By abuse of notation, we continue to write $T_n$
for $T_n(\omega)$ in such a case.

\begin{figure}[H]
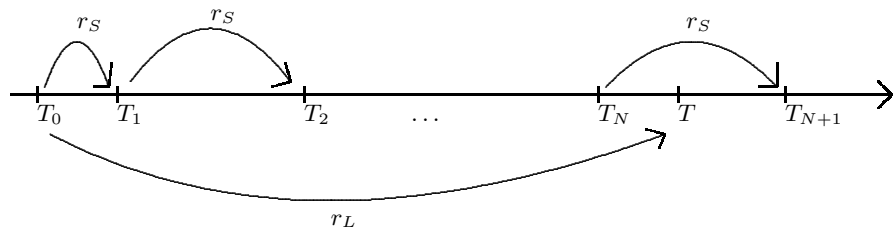
\small\caption{\small\textbf{Maturities of Short- and Long-Term Debt}}\label{fig:maturity}
\begin{center}
\hspace{-13cm}{ $ \drawline(0,50)(330,50) \drawline(10,47)(10,53)
\drawline(40,47)(40,53) \drawline(110,47)(110,53)
\drawline(220,47)(220,53) \drawline(250,47)(250,53)
\drawline(290,47)(290,53) \tagcurve(37,53, 25,70, 13,53)
\tagcurve(13,53, 25,70, 37,53)
\drawline(37,53)(31,53)\drawline(37,53)(38,62) \tagcurve(105,55,
75,75, 45,55) \tagcurve(45,55, 75,75, 105,55)
\drawline(105,55)(98,53)\drawline(105,55)(103,62) \tagcurve(223,53,
255,70, 287,53) \tagcurve(287,53, 255,70, 223,53)
\drawline(287,53)(280,53)\drawline(287,53)(287,62) {\tagcurve(15,35,
120,10, 245,35) \tagcurve(245,35, 120,10, 15,35)
\drawline(245,35)(238,37)\drawline(245,35)(243,30)}
\drawline(325,45)(330,50)\drawline(325,55)(330,50)
\put(10,40){$T_0$}\put(25,75){$r_S$}\put(40,40){$T_1$}\put(75,78){$r_S$}\put(110,40){$T_2$}
\put(150,40){$\ldots$}\put(253,75){$r_S$}\put(120,1){{$r_L$}}\put(220,40){$T_N$}\put(250,40){{$T$}}
\put(290,40){$T_{N+1}$} $}
\end{center}
\end{figure}

At the terminal time $T$, the representative short-term creditor
faces the insolvency risk that the firm may not pay back her
funding, and her value function at the terminal time is\footnote{The
probability of the insolvency time $\tau^{Ins}$ equal to the
terminal time $T$ is zero, so at the terminal time $T$ the firm only
faces the insolvency risk stemming from the final workout of the
firm's risky project. For this reason the recovery rate $R$ at time
$T$ is redefined as $R_T=\min\{1,X_T/(1+l_T)\}$.}
\begin{equation}\label{R_T}
U(T,x)=R_T=\min\left\{1,\frac{x}{1+l_T}\right\},
\end{equation}
which represents the insolvency risk stemming from the final workout
of the firm fundamental as in \cite{Merton}.

During the last time period $(T_N,T)$, all of the creditors are
locked in, so there is no rollover risk, and the representative
short-term creditor only faces the insolvency risk with the
associated recovery risk. Her value function for $t\in(T_N,T)$ is
\begin{align}\label{DynamicProgrammingEquation0}
U(t,x)=&\ \mathbf{E}_t^{x}\left\{\mathbf{1}_{\{t\leq \tau^{Ins}<T\}}
e^{-r(\tau^{Ins}-t)}\cdot e^{r_S(\tau^{Ins}-t)}R_{\tau^{Ins}}\right.\nonumber\\
&+ \left.\mathbf{1}_{\{ \tau^{Ins}\geq T\}}e^{-r(T-t)}\cdot
e^{r_S(T-t)}U(T,X_T)\right\},
\end{align}
where the first term in the bracket captures the insolvency risk
from the firm fundamental falling below the insolvency barrier
$D_t^{Ins}$ during the time period $(t,T)$, and the second term
captures the insolvency risk from the final workout of the firm's
risky project at time $T$. Hence,
(\ref{DynamicProgrammingEquation0}) represents the insolvency risk
due to the deterioration of the firm fundamental as in
\cite{BlackCox}.

To determine the value function at $t=T_N$ we take a closer look at
the rollover decision problem. At the rollover date $T_N$, if the
firm survives a debt run, the representative short-term creditor
will compare the expected return from rolling over her funding with
the expected market return, and will choose whatever results in a
higher return for her. If the firm defaults due to a debt run, she
will receive the recovery value $R_{T_N}$ in any case, regardless of
whether she decides to roll over her funding or not. Hence, the
creditor can only make her rollover decision conditional on the firm
surviving the current debt run. Therefore, the value function given
in equation (\ref{DynamicProgrammingEquation0}) also describes her
discounted expected return at time $t=T_N$ for the remaining time
period~$(T_N,T)$.

In general, during the time period $[T_n,T_{n+1})$ for
$n=0,1,\ldots,N-1$, the representative short-term creditor is
exposed not only to the insolvency risk arising from the
deterioration of the firm fundamental in the period $[T_n,T_{n+1})$
but also to the rollover risk caused by other creditors' rollover
decisions at time $T_{n+1}$. Table~\ref{table aggregate payoff}
summarizes her payoff at maturity $T_{n+1}$.

\begin{table}[H]\small\caption{\small\textbf{Representative creditor's aggregate payoff from $T_n$ to $T_{n+1}$}}
\label{table aggregate payoff}
\begin{center}
\begin{tabular}{cccc}\toprule
 Decision        &Solvency in $[T_n,T_{n+1}]$, & Solvency in $[T_n,T_{n+1}]$, &Insolvency in $[T_n,T_{n+1}]$. \\
at $T_{n+1}$ & no default due to run at $T_{n+1}$. & default due to
run at $T_{n+1}$. &
\\\midrule
Run       & $ e^{r_S(T_{n+1}-T_n)}\cdot1$ & $
e^{r_S(T_{n+1}-T_n)}\cdot R_{T_{n+1}}$ & $e^{r_S(\tau^{Ins}-T_n)}R_{\tau^{Ins}}$ \\
Rollover  & $ e^{r_S(T_{n+1}-T_n)}\cdot U(T_{n+1},X_{T_{n+1}})$ &
$e^{r_S(T_{n+1}-T_n)}\cdot R_{T_{n+1}}$    &
       $e^{r_S(\tau^{Ins}-T_n)}R_{\tau^{Ins}}$\\\bottomrule
\end{tabular}
\end{center}
\end{table}

At maturity $T_{n+1}$ if there is no default, the representative
short-term creditor either withdraws her funding to get
$e^{r_S(T_{n+1}-T_n)}\cdot1$ or renews her contract to receive
$e^{r_S(T_{n+1}-T_n)}\cdot U(T_{n+1},X_{T_{n+1}})$. If the firm
defaults due to a debt run at time $T_{n+1}$, the creditor just gets
the fraction $R_{T_{n+1}}$ of her funding $e^{r_S(T_{n+1}-T_n)}$
back. Since the creditor believes that the firm survives a debt run
at time $T_{n+1}$ with probability $\theta(X_{T_{n+1}})$, her
discounted expected return at time $t\in[T_n,T_{n+1})$ can be
described by the following value function
\begin{align}\label{DynamicProgrammingEquation}
U(t,x)=&\ \mathbf{E}_t^{x}\left\{\mathbf{1}_{\{t\leq
\tau^{Ins}<T_{n+1}\}}
 e^{(r_S-r)(\tau^{Ins}-t)}R_{\tau^{Ins}}+\mathbf{1}_{\{ \tau^{Ins}\geq
T_{n+1}\}}
e^{(r_S-r)(T_{n+1}-t)}\cdot \right.\nonumber\\
&\hspace{10mm}
\left.\cdot\left[\theta(X_{T_{n+1}})\max\left\{1,U(T_{n+1},X_{T_{n+1}})\right\}+
(1-\theta(X_{T_{n+1}}))R_{T_{n+1}}\right]\right\}.
\end{align}
The first term on the right hand side captures the insolvency risk
within the time period $[t,T_{n+1})$, whereas the second term
captures the rollover risk at time $T_{n+1}$ as well as the
insolvency and rollover risks in $[T_{n+1},T]$. Therefore, the first
term in the second line of (\ref{DynamicProgrammingEquation})
represents the future rollover risk (and insolvency risk) as in
\cite{HeXiongRun}, and the second term in the second line of
(\ref{DynamicProgrammingEquation}) represents the current rollover
risk as in \cite{MorrisShin}. Note that if the maturities are
random, then both (\ref{DynamicProgrammingEquation0}) and
(\ref{DynamicProgrammingEquation}) are understood as conditioning on
one realization of $\{T_n(\omega)\}_{n\geq 1}$ so that $T_N<T\leq
T_{N+1}$.

The dynamic programming equations
(\ref{DynamicProgrammingEquation0}) and
(\ref{DynamicProgrammingEquation}) for the value function $U(t,x)$
are the drivers to determine the debt run barrier in our model,
which will be discussed later. By the Feynman-Kac formula, we have
the following partial differential equation (PDE) representation for
the value function $U(t,x)$.

\begin{proposition}\label{Feynman-Kac1} Suppose Assumptions \ref{Assumption1}, \ref{Assumption2}, and \ref{Assumption3} are satisfied.
For $n=0,1,\ldots,N-1$, let $W_n(t,x)$ be the unique solution to the
following PDE Dirichlet problem on $[T_{n},T_{n+1}]\times[\beta
l_t,\infty)$
\begin{equation}\label{PenalizedEqu}
\left\{
\begin{array}{l}
\frac{\partial W_n}{\partial
t}+\mathcal{L}W_n+(r_S-r)W_n=0 \\[0.4cm]
W_n(t,\beta l_t)=\alpha\beta l_t/(1+l_t)\\[0.4cm]
W_n(T_{n+1},x)=\theta(x)\max\left\{1,W_{n+1}(T_{n+1},x)\right\}+(1-\theta(x)) \alpha x/(1+l_{T_{n+1}}).%
\end{array}%
\right.
\end{equation}%
For $n=N$, let $W_N(t,x)$ be the unique solution to the following
Dirichlet problem on $[T_{N},T]\times[\beta l_t,\infty)$
\begin{equation}\label{PenalizedEqu2}
\left\{
\begin{array}{l}
\frac{\partial W_N}{\partial
t}+\mathcal{L}W_N+(r_S-r)W_N=0 \\[0.4cm]
W_N(t,\beta l_t)=\alpha \beta l_t/(1+l_t)\\[0.4cm]
W_N(T,x)=\min \left\{ 1, x/(1+l_T)\right\},%
\end{array}%
\right.
\end{equation}%
where $\mathcal{L}$ is the infinitesimal generator for the ratio
process $X$ given by
$$\mathcal{L}=\frac{1}{2}\sigma^2x^2\frac{\partial^2}{\partial x^2}
+(r_V-r_S)x\frac{\partial}{\partial x}.$$ Then the value function
$U(t,x)$ is given by concatenating $W_n(t,x)$ together
$$U(t,x)=W_{n}(t,x)\ \ \text{for}\ \
t\in[T_{n},T_{n+1}).$$
\end{proposition}

Based on the Green's function technique, we further have the
following analytical representation for the value function
$U(t,\bar{x})$ where $\bar{x}=x/ (\beta l_t)$.

\begin{proposition}\label{prop Green's representation} (Green's Representation)
For $n=0,1,\ldots,N$, denote $P_n$ and $\Phi_n$ respectively as the
boundary condition and the terminal condition of the corresponding
PDE for the value function $U(t,\bar{x})$ on $[T_{n},T_{n+1})$,
where $T_{N+1}:=T$ for convenience. Then
\begin{equation}\label{GreenFuntion}
U(t,\bar{x})=\int_1^{\infty}P_n(\xi)\mathbb{G}(t,\bar{x};T_{n+1},\xi)d\xi+
\frac12\sigma^2\int_t^{T_{n+1}}\Phi_n(\eta)\frac{\partial}{\partial\xi}\left.\left\{\xi^2\mathbb{G}(t,\bar{x};\eta,\xi)\right\}\right|_{\xi=1}d\eta
\end{equation}
on $[T_{n},T_{n+1})$, where $\mathbb{G}(t,\bar{x};\eta,\xi)$ is the
Green's function for the operator $\mathcal{L}^v$ defined as
$$\mathcal{L}^v=\frac{\partial}{\partial t}+\frac12\sigma^2\bar{x}^2\frac{\partial^2}{\partial\bar{x}^2}+
(r_V-r_L)\bar{x}\frac{\partial}{\partial\bar{x}}+(r_S-r)$$ on the
domain $[T_{n},T_{n+1}]\times[1,\infty)$ given by
\begin{align*}
\mathbb{G}(t,\bar{x};T_{n+1},\xi)=&\
\frac{e^{(r_S-r)(T_{n+1}-t)}}{\xi\sigma\sqrt{2\pi(T_{n+1}-t)}}
\exp\left\{-\frac{\left[\log\frac{\bar{x}}{\xi}+(r_V-r_L-\frac12\sigma^2)
(T_{n+1}-t)\right]^2}{2\sigma^2(T_{n+1}-t)}\right\}\\
&\times
\left[1-\exp\left\{\frac{2\log\frac{1}{\xi}\log\bar{x}}{\sigma^2(T_{n+1}-t)}\right\}\right].
\end{align*}
\end{proposition}

\begin{proof} See Appendix \ref{app A}.
\end{proof}

\section{Threshold Strategies of Debt Run Model}\label{sec barriers}

In this section, we use the dynamic programming equations
(\ref{DynamicProgrammingEquation0}) and
(\ref{DynamicProgrammingEquation}) to determine the debt run barrier
as well as the illiquidity barrier for the \emph{representative
short-term creditor}. Our main objective is to show the monotonic
relationship among the debt run barrier, the illiquidity barrier,
and the exogenously given insolvency barrier.

\subsection{Discrete Tenor Structure: Revisit of \cite{LiangLutkebohmertXiao}}\label{Remark}

In this subsection, we extend the main results in
\cite{LiangLutkebohmertXiao} to our general setup.
\cite{LiangLutkebohmertXiao} show that there exists a threshold,
called the \emph{debt run barrier} such that the representative
short-term creditor will withdraw her funding whenever the firm
fundamental falls below this barrier at a rollover date. The debt
run barrier is only a finite sequence of numbers, since the creditor
only has a finite number of rollover dates to decide whether to run
or not. In our general setting we define the debt run barrier
$D_{T_n}^{Run}$ for any $n=0, 1,\ldots, N$ as the critical asset
value such that the representative short-term creditor is
indifferent in terms of running or rolling over her debt, i.e., it
is defined via the unique value $x_{T_n}^{*}$ such that $1=U(T_n,
x^*_{T_n})$ in the maximum term in dynamic programming equation
(\ref{DynamicProgrammingEquation}). The debt run barrier
$D_{T_n}^{Run}$ is then determined by
$$D_{T_n}^{Run}=x^*_{T_{n}}S_{T_n}=x^*_{T_{n}}S_0e^{r_ST_{n}},
\ \ \ \text{for}\ n=0,1,\ldots, N.$$ Since such a debt run barrier
is determined by the \emph{representative short-term creditor}, it
is actually the debt run barrier for \emph{all of the short-term
creditors}, who will then run on the firm if $V_{T_n}\leq
D_{T_n}^{Run}$. Although the representative creditor (so every
short-term creditor) holds the belief that $\xi$ proportion of them
will run on the firm at each date $T_n$, they will actually run at
the same time based on such a debt run barrier strategy. In fact,
such kind of debt run barrier strategy also appears in dynamic debt
run models such as \cite{MorrisShin}, \cite{HeXiongRun} and
\cite{LiangLutkebohmertXiao}.

In the following, we show that such a debt run barrier always
dominates the insolvency barrier. Note that the value function
$U(t,x)$ is obviously increasing with respect to $x$, and when the
firm goes bankrupt due to insolvency at a rollover date
$T_n=\tau^{Ins}$, the value function is
$$U(T_n,\beta l_{T_n})=R_{T_n}=\alpha \beta l_{T_n}/(1+l_{T_n}).$$
Due to Assumption \ref{Assumption3} we have $\beta\leq
(1/l_{T_n}+1)$ so that $U(T_n,\beta l_{T_n})\leq
1=U(T_n,x^*_{T_n})$. Hence we have $\beta l_{T_n}\leq x^*_{T_n}$.
This means the insolvency barrier $D_{t}^{Ins}$ at any rollover date
$t=T_n$ is dominated by the debt run barrier, i.e., $
D_{T_n}^{Ins}\leq D_{T_n}^{Run}$ for $n=0,1,\ldots,N.$ Note that
this dominance always holds in \cite{LiangLutkebohmertXiao}, since
the recovery rate $R_t$ is
assumed to be zero therein.\\

{A debt run does not necessarily trigger a default, for example in
the case where the firm can raise enough funding to pay off its
maturing short-term debt. The firm will survive the debt run at the
rollover date $T_n$, if $\psi V_{T_n}\geq S_{T_n}$ conditional on
$V_{T_n}\geq D_{T_n}^{Ins}$. Motivated by this observation we
introduce a third barrier, which we call an \emph{illiquidity
barrier} $D^{Ill}$, and which is defined as follows
\begin{equation}\label{equ illiquidity barrier}
D_{T_{n}}^{Ill}=\min\left\{D_{T_{n}}^{Run}, \max\{
S_{T_n}/\psi,D^{Ins}_{T_n}\}\right\}, \ \ \ \text{for}\
n=0,1,\ldots, N.
\end{equation}
Hence, an illiquidity default only occurs if there is a debt run and
the firm is not able to raise enough funds to pay off its maturing
short-term debt or not able to remain solvent at the debt run.
}


In the following, we show that the insolvency barrier $D_t^{Ins}$ at
$t=T_n$ is also dominated by the illiquidity barrier, i.e.,
$D_{T_n}^{Ins}\leq D_{T_n}^{Ill}$ for $n=0,1,\ldots,N.$ Indeed, we
always have
$$D_{T_n}^{Ill}\geq \min\{D^{Run}_{T_n},
D^{Ins}_{T_n}\}=D^{Ins}_{T_n}.$$

\begin{theorem}\label{Theorem1} Suppose that Assumptions \ref{Assumption1},
\ref{Assumption2}, and \ref{Assumption3} are satisfied. Then at any
maturity $T_n$, the debt run barrier is no less than the illiquidity
barrier, while the latter is no less than the insolvency barrier,
i.e.,
$$D_{T_n}^{Ins}\leq D_{T_n}^{Ill}\leq D_{T_n}^{Run}\ \ \ \text{for}\
n=0,1,\ldots,N.$$
\end{theorem}

Due to this relationship, the debt run barrier, the illiquidity
barrier and the insolvency barrier determine four possible scenarios
at each of the rollover dates
$T_n$ which are illustrated in the flowchart in Figure \ref{fig:flowchart}.\\

\centerline{[Insert Figure \ref{fig:flowchart} here.]}\vspace{0.4cm}

Figure \ref{fig: simulation discrete} shows different scenarios in
our debt run model with the discrete tenor structure for three
simulated asset value paths. Here we assume $\xi$ is uniformly
distributed, and $N=4$ rollover dates at times $t=2,4,6,$ and $8$.
The dotted line shows the debt run barrier, the dashed line the
illiquidity barrier, and the solid line the insolvency barrier. Note
that in this discrete setting, the debt run barrier and the
illiquidity barrier are not continuous functions. They consist only
of the marked points. The black asset value path falls below the
insolvency barrier shortly before time $t=4.$ At the rollover date
$t=2$ prior to this time, the asset value is larger than the debt
run barrier. Hence, in this simulation the firm will default shortly
before time $t=4$ due to insolvency. The dark black path falls below
the illiquidity barrier at the third rollover date at time $t=6$,and
before time $t=6$ it always stays above the insolvency barrier.
Thus, in this simulation the firm defaults due to illiquidity at
time $t=6$. Finally, the grey path shows a scenario where a debt run
occurs at the last rollover date $t=8$. At that time, however, the
asset value is still larger than the illiquidity barrier, meaning
that the firm is able to raise enough funds to pay off its
short-term creditors. Hence,
the firm survives the debt run.\\

\centerline{[Insert Figure \ref{fig: simulation discrete} here.]}

\subsection{Staggered Tenor Structure}\label{SubSecMaturity}

In \cite{LiangLutkebohmertXiao} it is assumed that short-term debt
rollover dates are given by a deterministic sequence of numbers and
that they are the same for all short-term creditors. This assumption
is rather restrictive. The firm is highly exposed to rollover risk
in such a setting where all short-term funding expires at the same
date. In practice, however, firms tend to spread out their debt
expirations across time to reduce their exposure to liquidity risk.
In this section, we introduce a more flexible debt maturity
structure. Among others, \cite{Leland,Leland2} and \cite{LelandToft}
introduced the so-called \emph{staggered maturity structure} to
capture this fact. The idea is to use the arrival times of a Poisson
process to model the maturities of short-term debt. In other words,
the duration of short-term debt $T_1-T_0, T_2-T_1, \ldots$ has an
exponential distribution. While the random duration assumption
appears different from the standard debt contract with a
predetermined maturity, it captures the staggered debt maturity
structure of a typical firm. For the application of such a
\emph{Poisson maturity structure} in the literature of debt runs, we
refer to the recent work by \cite{HeXiongRun}.

The crucial parameter under the aforementioned Poisson maturity
structure framework is the intensity $\lambda$. Its inverse
$1/\lambda$ can be interpreted as the average duration of short-term
debt. We consider a \emph{Cox maturity structure}, meaning that the
maturity of short-term debt follows a more general and flexible Cox
process. Recall that a Cox process is a generalization of Poisson
processes in which the intensity is allowed to be random but in such
a way that if we condition on a particular realization
$\lambda_t(\omega)$ of the intensity, the process becomes an
inhomogeneous Poisson process with intensity $\lambda_t(\omega)$.
The economic intuition of using the Cox maturity structure is that
the average duration of the short-term debt that the firm issues
should depend on some time-dependent economic factors such as the
firm fundamental $V$, the ratio $X$ of the firm fundamental over the
short-term debt, or even some underlying states of the economy. In
the following we therefore assume that the average maturity is a
function of the ratio process $X$.

We construct the short-term debt maturities $\{T_{n}\}_{n\geq 1}$ by
so-called \emph{canonical construction}. Let $\{E_n\}_{n\geq 1}$ be
a sequence of independent identically distributed (i.i.d.)
exponential random variables on some complete probability space
$(\tilde{\Omega},\tilde{\mathcal{F}},\tilde{\mathbf{P}})$, and
define the enlarged probability space by
$$\bar{\Omega}=\Omega\times\tilde{\Omega},\ \ \ \bar{\mathcal{G}}=\mathcal{F}\otimes\tilde{\mathcal{F}},
\ \ \ \text{and}\ \
\mathbf{Q}=\mathbf{P}\otimes\tilde{\mathbf{P}}.$$ We assume the
intensity has the form $\lambda_t=g(X_t)$, where $g:
(0,\infty)\rightarrow(0,\infty)$ is a smooth function with compact
support. Then the maturities of short-term debt are constructed
recursively as
$$
T_0=0 \quad \mbox{and}\quad T_n=\inf\left\{t\geq T_{n-1}:
\int_{T_{n-1}}^tg(X_s)ds\geq E_n\right\},\quad \mbox{for}\, n\geq 1.
$$

We summarize the
above construction in the following assumption.

\begin{assumption}\label{Assumption4} (Cox maturity structure) The maturities of the short-term debt $\{T_{n}\}_{n\geq 1}$ are the
arrival times of a Cox process with intensity $g(X_t)$.
\end{assumption}

Under the Cox maturity structure, we still employ the representative
short-term creditor's dynamic programming equations
(\ref{DynamicProgrammingEquation0}) and
(\ref{DynamicProgrammingEquation}) to determine her value function
$U(t,x)$. Letting the ratio process start from $X_t=x$ and the
short-term debt maturities start from $T_0=t$, we synthesize the
dynamic programming equations (\ref{DynamicProgrammingEquation0})
and (\ref{DynamicProgrammingEquation}) into the following succinct
form on the event $\{{T_1>t}\}$:
\begin{align}\label{DynamicProgrammingEquationMaturity}
&\ U(t,x)\nonumber\\
=&\ \mathbf{E}_t^{x}\left\{\mathbf{1}_{\{t\leq
\tau^{Ins}<T\}}\left[\mathbf{1}_{\{t<T_{1}<\tau^{Ins}\}}
e^{(r_S-r)(T_{1}-t)}\left[\theta(X_{T_{1}})\max\left\{1,U(T_{1},X_{T_{1}})\right\}+
(1-\theta(X_{T_{1}}))R_{T_{1}}\right]\right.\right.\nonumber\\
& \ \ \ \ \ \ \ \ \ \ \ \ \ \ \ \ \ \ \ \ \ \ \ \
\left.+\mathbf{1}_{\{ T_{1}\geq \tau^{Ins}\}}
e^{(r_S-r)(\tau^{Ins}-t)}R_{\tau^{Ins}}\right]\nonumber\\
& + \mathbf{1}_{\{ \tau^{Ins}\geq T\}}\left[\mathbf{1}_{\{t<
T_{1}<T\}}
e^{(r_S-r)(T_{1}-t)}\left[\theta(X_{T_{1}})\max\left\{1,U(T_{1},X_{T_{1}})\right\}
+(1-\theta(X_{T_{1}}))R_{T_{1}}\right]\right.\nonumber\\
& \ \ \ \ \ \ \ \ \ \ \ \ \ \ \ \ \ \ \ \ \ \ \ \
\left.\left.+\mathbf{1}_{\{ T_{1}\geq T\}}
e^{(r_S-r)(T-t)}\min\left\{1, X_T/(1+l_T)\right\}\right]\right\}.
\end{align}

By using the distribution of the first arrival time $T_{1}$, and
applying the Feynman-Kac formula, we derive the following PDE
representation for the value function $U(t,x)$ under the Cox
maturity structure.

\begin{proposition}\label{Feynman-Kac3}
Suppose that Assumptions \ref{Assumption1}, \ref{Assumption2},
\ref{Assumption3} and \ref{Assumption4} are satisfied. Then the
value function $U(t,x)$ satisfies the following semi-linear PDE
Dirichlet problem on $[0,T]\times[\beta l_t,\infty)$:
\begin{equation}\label{PenalizedEquMaturity}
\left\{
\begin{array}{l}
\frac{\partial U}{\partial t}+\mathcal{L}U+(r_S-r-g(x))U\\[+0.4cm]
\ \ \ \ +g(x)
[\theta(x)\max\{1,U\}+(1-\theta(x)) \alpha x/(1+l_t)]=0\\[0.4cm]
U(t,\beta l_t)=\alpha \beta l_t/(1+l_t)\\[0.4cm]
U(T,x)=\min \left\{ 1,  x/(1+l_T)\right\}.%
\end{array}%
\right.
\end{equation}%
\end{proposition}
\begin{proof} See Appendix \ref{app E}.
\end{proof}\\

In Appendix \ref{app B}, we provide a numerical algorithm to
approximate the solution of the above PDE
(\ref{PenalizedEquMaturity}). In the rest of this section, we show
that PDE (\ref{PenalizedEquMaturity}) implies a unique threshold for
the representative short-term creditor, i.e., there exists a unique
\emph{debt run barrier} $D_t^{Run}$ such that she will run on the
firm whenever both the firm's asset value falls below such a barrier
and her contract expires at some maturity $T_n$. Thus, the debt run
time in our model is characterized endogenously by the following
first-passage-time
\begin{equation*}
\tau^{Run}=\inf\{T_n: X_{T_n}\leq x^*(T_n)\}\wedge T,
\end{equation*}
where $x^*(t)$ is the threshold we shall derive in the remainder of
this section. Recall that $X_t=V_t/S_t$ is the ratio of the firm
fundamental over the short-term debt, so the debt run barrier
$D_t^{Run}$ is given as
$$D_t^{Run}=x^{*}(t)S_t=x^*(t)S_0e^{r_St}.$$

We derive a free-boundary problem to determine first the threshold
$x^*(t)$ and secondly the debt run barrier $D_t^{Run}$ based on the
semi-linear PDE (\ref{PenalizedEquMaturity}).
\begin{itemize}
\item[(i)] If $x>x^*(t)$, the representative short-term creditor will keep
lending her money to the firm because either the debt is not due yet
or if the debt is due she decides to roll over her funding. Her
value function $U(t,x)>1$, and (\ref{PenalizedEquMaturity}) reduces
to
\begin{equation}\label{ConEqu}
\frac{\partial U}{\partial
t}+\mathcal{L}U+(r_S-r-g(x))U+g(x)\theta(x)U+g(x)(1-\theta(x))
\alpha x/(1+l_t)=0.
\end{equation}
The third term in the above equation represents the creditor's
premium of the return, the fourth term represents the expected
effect of the rollover risk if the creditor rolls over her funding,
and the last term represents the expected effect of recovery risk
associated with a potential debt run.
\item[(ii)] If $x< x^*(t)$, the representative short-term creditor will run on the firm
if the debt is due. Her value function $U(t,x)<1$, and
(\ref{PenalizedEquMaturity}) reduces to
\begin{equation}\label{StoEqu}
\frac{\partial U}{\partial
t}+\mathcal{L}U+(r_S-r-g(x))U+g(x)\theta(x)+g(x)(1-\theta(x)) \alpha
x/(1+l_t)=0.
\end{equation}
While the third term and the last term in (\ref{StoEqu}) have the
same meanings as those in (\ref{ConEqu}), the fourth term captures
the expected effect of rollover risk from the representative
short-term creditor's own run.
\item[(iii)] Finally, by the continuity of $U(t,x)$, the creditor's value function $U(t,x)$
at the threshold $x^*(t)$ should be equal to $1$, and the following
smooth-pasting condition should be satisfied
$$U_{x+0}(t,x^*(t))=U_{x-0}(t,x^*(t)).$$
\end{itemize}

In summary, we obtain the following two-phase free-boundary problem
to  determine the threshold (i.e., the debt run barrier) of the
representative short-term creditor (so every short-term creditor).
\begin{proposition}\label{TheoremForThreshold}
Suppose that Assumptions \ref{Assumption1}, \ref{Assumption2},
\ref{Assumption3}, and \ref{Assumption4} are satisfied. Then the
debt run barrier $D^{Run}_t$ is given by
$$D^{Run}_t=x^*(t)S_0e^{r_St},$$
where $x^*(t)$ is the free-boundary of the following two-phase
free-boundary problem
\begin{equation}\label{TwoPhaseProblem}
\left\{
\begin{array}{ll}
\frac{\partial U}{\partial
t}+\mathcal{L}U+(r_S-r-g(x))U+g(x)\theta(x)U+g(x)(1-\theta(x))
\alpha x/(1+l_t)=0, & \text{for}
\ \ x>x^*(t),\\[0.4cm]
U(t,x)>1,&\text{for}\ \ x>x^*(t),\\[0.4cm]
U(t,x)=1,&\text{for}\ \ x=x^*(t),\\[0.4cm]
U_x(t,x)\ \text{is\ continuous},&\text{for}\ \ x=x^*(t),\\[0.4cm]
\frac{\partial U}{\partial
t}+\mathcal{L}U+(r_S-r-g(x))U+g(x)\theta(x)+g(x)(1-\theta(x))
\alpha x/(1+l_t)=0, & \text{for} \ \ \beta l_t<x<x^*(t),\\[0.4cm]
U(t,x)<1,&\text{for}\ \ \beta l_t<x<x^*(t),\\[0.4cm]
U(t,x)=\alpha \beta l_t/(1+l_t),&\text{for}\ \ x=\beta l_t,\\[0.4cm]
U(T,x)=\min\{1, x/(1+l_T)\}.
\end{array}%
\right.
\end{equation}%
\end{proposition}

\begin{proof} We only need to prove the smooth-pasting condition, which is straightforward since PDE
(\ref{PenalizedEquMaturity}) admits a unique classical solution.
\end{proof}\\

{Similar to the case of the discrete tenor structure in section
\ref{Remark}, a debt run does not necessarily trigger the firm's
default. The firm will not default due to a debt run if the firm can
raise enough funding to pay off its short-term creditors who run on
the firm, and remain solvent at the debt run, i.e., if $\psi
V_{T_n}\geq S_{T_n}$ conditional on $V_{T_n}\geq D^{Ins}_{T_n}$.
Therefore, we define the firm's \emph{illiquidity barrier} as
$$D_{t}^{Ill}=\min\left\{D_{t}^{Run},
\max\{S_{t}/\psi,D_t^{Ins}\}\right\} \ \ \ \text{for}\ t\in[0,T).$$
However, such a barrier only acts at a sequence of Cox arrival times
$\{T_n\}_{n\geq 1}$. At each infinitesimal time interval $[t,t+dt]$
with probability $g(X_t) dt$, the short-term debt matures, and an
illiquidity default will happen if $V_t\leq D_t^{Ill}$. With
probability $1-g(X_t)dt$, the short-term debt does not mature yet,
so all of the short-term creditors are locked in, and an illiquidity
default will not occur even if $V_t\leq D_T^{Ill}$. We have a
similar relationship among the barriers as in the case of the
discrete tenor structure.}

\begin{theorem}\label{Theorem2} Suppose that Assumptions \ref{Assumption1},
\ref{Assumption2}, \ref{Assumption3}, and \ref{Assumption4} are
satisfied. Then at any time $t\in[0,T)$, the debt run barrier is
greater than or equal to the illiquidity barrier, while the latter
is greater than or equal to the insolvency barrier
$$D_{t}^{Ins}\leq D_{t}^{Ill}\leq D_{t}^{Run}\ \ \ \text{for}\
t\in[0,T).$$ The above relationship gives us the following four
possible scenarios at any rollover date~$T_n.$
 \begin{itemize}
 \item[(i)] $V_{T_n}\leq D_{T_n}^{Ins}$: Default due to insolvency;
 \item[(ii)] $D_{T_n}^{Ins}<V_{T_{n}}\leq D_{T_n}^{Ill}$: Debt run occurs
 and triggers a default due to illiquidity;
 \item[(iii)] $D_{T_n}^{Ill}<V_{T_{n}}\leq D_{T_n}^{Run}$: Debt run occurs,
 but no default caused by the run;
 \item[(iv)] $D_{T_n}^{Run}< V_{T_{n}}$: The creditor rolls over
 to the next maturity $T_{n+1}$.
 \end{itemize}
\end{theorem}

\begin{proof} The proof is essentially the same as the proof for
Theorem \ref{Theorem1}, so we omit it.
\end{proof}\\

Figure \ref{fig: simulation continuous} illustrates different
scenarios in our debt run model with the staggered tenor structure
of short-term debt and the uniform distribution on the proportion of
short-term creditors not rolling over their funding at each rollover
date. Here the intensity of the Cox process is chosen to be
$g(x)=0.4$. The dotted line shows the debt run barrier, the dashed
line the illiquidity barrier and the solid line the insolvency
barrier, all of which are continuous functions in this model
setting. The marked times $T_1$, $T_2$, and $T_3$ are one
realization of the arrival times of the Cox process, which are
smaller than the final date $T=10$. At the first rollover date $T_1$
all three asset value paths are above the debt run barrier. Hence,
all of the creditors roll over their contracts. The black asset
value path is above all of the three barriers at $T_2$, but falls
below the insolvency barrier before the third rollover date $T_3$.
Hence, the firm will default at that time point before $T_3$. At the
second rollover date $T_2$, the grey path falls below the debt run
barrier but is still above the illiquidity barrier. This means that
a debt run occurs at that date, but the firm is able to pay off its
creditors and survives. At the last rollover date $T_3$ the dark
black path is below the illiquidity barrier, which means that a debt
run occurs and the debt run actually triggers an illiquidity
default. Note that all three paths fall below the debt run and the
illiquidity barriers already much earlier in time. However, as these
times are not rollover dates for the representative short-term
creditor, she cannot withdraw her funding at these dates.

The figure also illustrates the relation between different barriers,
which has been theoretically proved in Theorem \ref{Theorem2}; the
debt run barrier is always greater than or equal to the illiquidity
barrier, which in turn is always greater than or equal to the
insolvency barrier.\\

\centerline{[Insert Figure \ref{fig: simulation continuous}
here.]}\vspace{0.4cm}

Finally, we show in Figure \ref{fig:distribution} the debt run
barrier $D_t^{Run}$ with different assumptions on the random
variable $\xi$, the representative short-term creditor's belief on
the proportion of short-term creditors not rolling over their
funding. The top line shows the debt run barrier when $\xi$ is
\emph{uniformly distributed}: $f(\cdot|X_{T_n})=1$. From the next
top line to the bottom line, they correspond to the debt run barrier
with $\xi$ following \emph{normal distribution} truncated on $[0,1]$
with mean $\mu=0.5$ and variance $\sigma^2=1,1/3,1/6,1/12$ on $T_n$.
That is, $\xi$ has the conditional density
$$f(x|X_{T_n})=f(x)=\displaystyle\frac{\frac{1}{\sigma}\phi(\frac{x-\mu}{\sigma})}{\Phi(\frac{1-\mu}{\sigma})-\Phi(\frac{-\mu}{\sigma})},$$
where $\phi$ and $\Phi$ are the density function and the cumulative
distribution function of a normal random variable $N(\mu,\sigma^2)$
respectively. The top line and the bottom line have the same mean
$0.5$ and the same variance $1/12$, which illustrates that the
uniform distribution assumption is more conservative for the
creditors compared to the normal distribution assumption. On the
other hand, for the normally distributed $\xi$ with the same mean,
the larger the variance is, the higher the debt run barrier. \\

\centerline{[Insert Figure \ref{fig:distribution} here.]}

\subsection{Comparison of the Discrete and the Staggered Tenor Structures}

{This section compares different debt tenor structures, i.e., the
discrete and the staggered tenor structures. We first calculate the
survival probabilities under both tenor structures. For the
calculation of the survival probability under the discrete tenor
structure setting, we refer to \cite{LiangLutkebohmertXiao}. The
calculation of the survival probability under the staggered tenor
structure is tricky, as the number of Cox arrival times happening
during the time interval $[0,T]$ is random. Inspired by the
recursive formula (\ref{DynamicProgrammingEquationMaturity}) for the
calculation of the value function, we also calculate the survival
probability in a recursive way. Let $P(t,x)$ be the corresponding
survival probability at time $t$ given the current ratio $X_t=x$.
Then on the event $\{T_1>t\}$,
\begin{align*}
P(t,x)=&\ \mathbf{E}_t^x\left\{\mathbf{1}_{\{t\leq\tau^{Ins}<T\}}\left[\mathbf{1}_{\{t<T_1<\tau^{Ins}\}}\mathbf{1}_{\{X_{T_1}\geq x^{Ill}(T_1)\}}P(T_1,X_{T_1})\right]\right.\\
&\ +\left.\mathbf{1}_{\{\tau^{Ins}\geq
T\}}\left[\mathbf{1}_{\{t<T_1<T\}}\mathbf{1}_{\{X_{T_1}\geq
x^{Ill}(T_1)\}}P(T_1,X_{T_1})+\mathbf{1}_{\{T_1\geq
T\}}\right]\right\},
\end{align*}
where
$x^{Ill}(T_1)=D_{T_1}^{Ill}/S_{T_1}=\min\{x^*(T_1),\max\{1/\psi,\beta
l_{T_1}\} \}$.
By Lemma \ref{basiclemma}, the survival probability $P(t,x)$ can be
calculated as
\begin{align*}
P(t,x)=&\ \mathbf{E}_t^{x}\left\{\int_t^{\tau^{Ins}\wedge T}
e^{-\int_t^{s}g(X_u)du}g(X_s)\left[\mathbf{1}_{\{X_{s}\geq x^{Ill}(s)\}}P(s,X_{s})\right]ds\right.\\
&\ +\left.\mathbf{1}_{\{\tau^{Ins}\geq
T\}}e^{-\int_t^{T}g(X_u)du}\right\}.
\end{align*}}

{Therefore, the Feynman-Kac formula gives the following semilinear
PDE representation for the survival probability $P(t,x)$:
\begin{equation}\label{PenalizedEquMaturity2}
\left\{
\begin{array}{l}
\frac{\partial P}{\partial t}+\mathcal{L}P-g(x)P+g(x)
\left[\mathbf{1}_{\{x\geq x^{Ill}(t)\}}P\right]=0\\[+0.4cm]
P(t,\beta l_t)=0\\[0.4cm]
P(T,x)=1.%
\end{array}%
\right.
\end{equation}%
Its solution can be numerically approximated in a similar way to the
numerical approximation for (\ref{PenalizedEquMaturity}) in Appendix
\ref{app B}. Given the survival probability $P(t,x)$, the default
probability can then be calculated as $1-P(t,x)$.\\}

In Figure \ref{fig:PDs discrete and staggered} we show the default
probabilities under both the discrete and the staggered tenor
structures with uniformly distributed $\xi$. The dashed-dotted line
is the default probability without taking rollover risk into
account, which corresponds to the default probability in the setting
of \cite{BlackCox}. Hence, the areas between the dashed-dotted line
and the other two lines represent the rollover risks induced by debt
runs under the staggered tenor structure and under the discrete
tenor structure, resp. The results show that the default probability
is increasing with increasing volatility as the asset value becomes
more risky. Furthermore, the figure supports the intuition that
replacing the discrete tenor structure by a staggered tenor
structure reduces liquidity risk.\\

\centerline{[Insert Figure \ref{fig:PDs discrete and staggered}
here.]}\vspace{0.4cm}

In Figures \ref{fig: PD11} and \ref{fig: PD5} we see a kink in the
default probability for the staggered tenor structure. This kink is
less noticeable in Figures \ref{fig: PD7} and \ref{fig: PD9} and
does not appear in the discrete tenor structure case. This indicates
that this feature is due to the different specifications in the
intensity function $g(x)$ of the Cox process {and the volatility
$\sigma$ of the firm fundamental}. For very low asset values
$\tau^{Ins}<T_1$ with high probability, the insolvency component
determines the profile of the default probability in the staggered
tenor structure case. For higher asset values $T_1$ can be smaller
than $\tau^{Ins}$. In this case there exists a critical asset value
such that creditors will very likely decide to withdraw when the
asset value is below this level and a run will most likely induce an
illiquidity default. This implies an almost flat default probability
for asset values below this critical level.
For higher asset values a debt run does not necessarily imply an illiquidity default and thus the default probability is monotonically decreasing for increasing asset value. The same feature is noticeable in Figure \ref{fig: PD6} which we will discuss below. \\


\centerline{[Insert Figure \ref{fig:PDs for different psi}
here.]}\vspace{0.4cm}

When creditors fear that the firm will be unable to repay their
debt, they will withdraw their funding simultaneously at a rollover
date and thereby, they might trigger an illiquidity induced default.
The key quantity  in our model that determines the creditors'
behavior is their beliefs on the survival probability of the debt
run $\theta$. When $\theta$ is close to one, creditors are
optimistic to get paid off. This is the case when either the firm's
asset value and the fire-sale rate are high or when the short-term
debt notional is very low. Figure \ref{fig:PDs for different psi}
shows
that the illiquidity component of the default probability
dramatically increases when fire-sale rate decreases. Thus,
creditors' might withdraw their funding because they have a very
pessimistic view on the firm's ability to repay them (low fire-sale
rate induces low $\theta$) although the firm's asset value might be
well above the insolvency barrier. This supports the idea that debt
runs can occur as a result of pure coordination failures where
$\theta$ can be interpreted as a coordination
parameter.\footnote{\cite{Arifovic2013} study how coordination
problems can affect the occurrence of bank runs in controlled
laboratory
environments.}\\


A natural question to ask is, what will happen with the discrete
tenor structure when the number of rollover dates increases to
infinity, meaning that creditors can decide to roll over or to
withdraw their funding at any time $t\in[0,T]$? Intuitively, one
would expect that with increasing rollover frequency, one should
approximate the staggered tenor structure model. However, there is
another important difference between the two debt tenor structures.
In the case of the discrete tenor structure we implicitly assume
that all creditors have the same rollover dates, whereas in the
staggered tenor structure model at each rollover date, corresponding
to a Cox arrival time, only a fraction of total debt is due.
Different short-term creditors hence have different rollover dates
in that situation.

In the following, we will first study the impact of the intensity
$\lambda_t=g(X_t)$ of the Cox process on the creditor's value
function. We assume the function $g(X_t)$ to be constant and thus
independent of the ratio process $X_t$. The intensity of the Cox
process not only specifies the creditor's rollover dates but also
affects the average duration of short-term debt. For $g(X_t)\equiv
g\in\mathbb{R}_+$ the average duration of debt is equal to $1/g$,
and in an infinitesimal time interval $[t,t+dt]$ a fraction $gdt$ of
debt is maturing. The larger is $g$, the more debt is maturing at
the same rollover date, and the larger is the rollover frequency of
short-term debt. Therefore, for large enough $g$ the staggered tenor
structure model and the discrete tenor structure model should result
in approximately the same value function $U(t,x)$ for the short-term
debt. This result is numerically validated in Figure~\ref{fig
comparison variable g}. The number of rollover dates in the discrete
tenor structure model is fixed at $N=1000$, and the intensity of the
Cox process in the staggered
tenor structure model varies from $g=0.2$ to $g=200$. This supports our previous discussion of the kink in the default probabilities for the staggered tenor structure visible in Figures \ref{fig: PD11} and \ref{fig: PD5} where $g(x)=0.2$. When increasing the intensity to $g(x)=0.4$ as in Figures \ref{fig: PD7} and \ref{fig: PD9} the above discussed effect becomes less prominent and the default probability in the staggered tenor structure and the discrete tenor structure case are much closer.\\

\centerline{[Insert Figure \ref{fig comparison variable g}
here.]}\vspace{0.4cm}

In section \ref{Remark}, we derived the debt run barrier for the
discrete tenor structure by determining the threshold ratio $x^*$
such that $U(t,x^{*}(t))=1$, i.e., the creditor is indifferent
between rolling over and withdrawing her funding. Similarly in
Proposition \ref{TheoremForThreshold}, we derived the debt run
barrier for the staggered tenor structure by solving the
free-boundary problem (\ref{TwoPhaseProblem}). Next, we will
investigate in Figure \ref{fig comparison variable g and N} the
impact of the Cox intensity $g$ and the rollover frequency $N$ on
the debt run barrier. The graphs show that the discrete tenor
structure model with high rollover frequency $N$ approximates the
staggered tenor structure model with large intensity $g(x)$.\\

\centerline{[Insert Figure \ref{fig comparison variable g and N}
here.]}

\section{Optimal Stochastic Control Formulation}\label{sec optimal control}

In this section we are concerned with whether the representative
short-term creditor's decision is optimal. Intuitively, since the
creditor's decision follows the DPP, her decision should be optimal.
The question then is, what is the corresponding optimal stochastic
control problem? To answer this question we first investigate the
case of the discrete tenor structure and then discuss the staggered
debt structure.

\subsection{Discrete Tenor Structure}
Let us first consider the case of the discrete tenor structure,
i.e., short-term debt maturities $\{T_{n}\}_{n\geq 1}$ are a
sequence of deterministic numbers. Recall that at each rollover date
$T_n$, the creditor believes that there is a probability
$(1-\theta(X_{T_n}))$ that the firm may default due to debt runs.
Let $T_{*}$ denote the time that the firm defaults due to a debt
run. Hence, $T_{*}$ is a random time taking value in
$\{T_{n}\}_{n\geq 1}$.

Let $\tau\in\{T_n\}_{n\geq 1}\backslash T_*$ be the time at which
the representative short-term creditor decides to withdraw her
funding and to run on the firm. This is an $\mathcal{F}_t$-stopping
time. We first consider the case $\{\tau^{Ins}< T\}$, i.e., the firm
fails due to insolvency before its project expires. If
$\tau<T_{*}\wedge \tau^{Ins}$, the creditor withdraws her funding
before an illiquidity or insolvency default happens. In this case
she will obtain the payoff
$$\mathbf{1}_{\{\tau<T_{*}\wedge \tau^{Ins}\}}e^{r_S \tau}.$$
If $T_{*}<\tau\wedge\tau^{Ins}$, the firm fails due to the debt run
before the creditor decides to withdraw her money and before an
insolvency happens. Hence, the creditor will obtain the payoff
$$\mathbf{1}_{\{T_{*}<\tau\wedge\tau^{Ins}\}}e^{r_ST_{*}}R_{T_{*}}.$$
Finally, if $\tau^{Ins}\leq T_{*}\wedge \tau$, the firm defaults due
to insolvency before the illiquidity default takes place and before
the creditor decides to withdraw her funding. Then, the creditor
will obtain the payoff
$$\mathbf{1}_{\{\tau^{Ins}\leq T_{*}\wedge \tau\}}e^{r_S\tau^{Ins}}R_{\tau^{Ins}}.$$
On the other hand, on the event $\{\tau^{Ins}\geq T\}$, i.e., no
insolvency happens before the project ends, the creditor will obtain
the payoff
$$\mathbf{1}_{\{\tau<T_{*}\wedge T\}}e^{r_S\tau}+
\mathbf{1}_{\{T_{*}<\tau\wedge T\}}e^{r_ST_{*}}R_{T_{*}}+
\mathbf{1}_{\{T\leq T_{*}\wedge \tau\}}e^{r_ST}\min\{1,
X_T/(1+l_T)\}.$$

Table \ref{tab payoff optimal control} summarizes the aggregate
payoff of the representative creditor.

\begin{table}[H]\small\caption{\small\textbf{Representative creditor's aggregate payoff}}\label{tab payoff optimal
control}
\begin{center}
\begin{tabular}{ccc}\toprule
Insolvency time $\tau^{Ins}$ &Decision time $\tau$ &   Payoff
\\\midrule
$\tau^{Ins}<T$& $\tau<T_{*}\wedge \tau^{Ins}$  &$e^{r_S \tau}\cdot 1$  \\
&$T_{*}<\tau\wedge\tau^{Ins}$   &$e^{r_ST_{*}}\cdot R_{T_{*}}$   \\
&$\tau^{Ins}\leq T_{*}\wedge \tau$  &$e^{r_S\tau^{Ins}}\cdot R_{\tau^{Ins}}$\\
\hline
$\tau^{Ins}\geq T$&$\tau<T_{*}\wedge \tau^{Ins}$  &$e^{r_S \tau}\cdot 1$  \\
&$T_{*}<\tau\wedge T$   &$e^{r_ST_{*}}\cdot R_{T_{*}}$   \\
& $T\leq T_{*}\wedge \tau$ &$e^{r_ST}\cdot \min\{1,
X_T/(1+l_T)\}$\\\bottomrule
\end{tabular}
\end{center}
\end{table}

For any $0\leq t\leq \hat{t}\leq T$ where $\hat{t}$ could be either
$\tau^{Ins}$ or $T$, we define the aggregate discounted payoff from
time $t$ to $\hat{t}$ as\footnote{Recall that $R_T=\min\{1,
X_T/(1+l_T)\}$ as defined in (\ref{R_T}).}
$$\mathcal{A}_{t,\hat{t}}=\mathbf{1}_{\{t<\tau<T_{*}\wedge \hat{t}\}}e^{(r_S-r)(\tau-t)}+
\mathbf{1}_{\{t<T_{*}<\tau\wedge \hat{t}
\}}e^{(r_S-r)(T_{*}-t)}R_{T_{*}}+ \mathbf{1}_{\{t<\hat{t}\leq
T_{*}\wedge \tau\}}e^{(r_S-r)(\hat{t}-t)} R_{\hat{t}}.$$ The
creditor will choose an optimal $\mathcal{F}_t$-stopping time to
maximize her expected payoff
\begin{align}\label{optimalstoppingtime1}
\sup_{\tau\in\{T_{n}\}_{n\geq 1}\backslash T_{*}}\mathbf{E}^x_0&
\left\{\mathbf{1}_{\{\tau^{Ins}<T\}}\cdot\mathcal{A}_{0,\tau^{Ins}}+\mathbf{1}_{\{
\tau^{Ins}\geq T\}}\cdot\mathcal{A}_{0,T}\right\}.
\end{align}
The following theorem states that the creditor's value function
$U(t,x)$ defined by the dynamic programming equations
(\ref{DynamicProgrammingEquation0}) and
(\ref{DynamicProgrammingEquation}) is indeed optimal.

\begin{theorem}\label{theoremoptimalstoppingtime1}
The value of the optimal stopping time problem
(\ref{optimalstoppingtime1}) is given by the value function $U(0,x)$
in the dynamic programming equation
(\ref{DynamicProgrammingEquation}). The optimal stopping time is
given by the earliest maturity date at which the firm fundamental
falls below the debt run barrier determined in section \ref{Remark},
i.e.,
$$\tau^{Run}=\inf\{T_n: V_{T_n}\leq D_{T_n}^{Run},\ n=0,1,\ldots, N\}\wedge T.$$
\end{theorem}

\begin{proof} See Appendix \ref{app C}.
\end{proof}

\subsection{Staggered Tenor Structure}

Next we consider the case of the staggered tenor structure, i.e.,
the maturities $\{T_{n}\}_{n\geq 1}$ are the arrival times of a Cox
process with intensity $g(X_t)$. Similar to the discrete tenor
structure, we define $T_{*}$ as the default time due to a debt run.
Since $T_{*}$ is chosen among the Cox arrival times
$\{T_{n}\}_{n\geq 1}$ with probability $(1-\theta(X_{T_n}))$, it is
well known that $T_{*}$ is the first arrival time of another Cox
process with intensity $g(X_t)(1-\theta(X_t))$. Let
$\tau\in\{T_{n}\}_{n\geq 1}\backslash T_{*}$ again denote the
rollover date at which the representative short-term creditor
decides to withdraw her funding and to run on the firm. This is a
$\mathcal{G}_t=\mathcal{F}_t\vee\mathcal{H}_t$-stopping time under
the staggered tenor structure with $\mathcal{H}_t=\sigma(\{T_1\leq
u\}:u\leq t)$, i.e., $\tau$ must be chosen from the arrival times of
the Cox process,

The representative short-term creditor will choose an optimal
$\mathcal{G}_t$-stopping time to maximize her expected payoff
\begin{align}\label{optimalstoppingtime2}
\sup_{\tau\in\{T_{n}\}_{n\geq 1}\backslash T_{*}}\mathbf{E}^x_0&
\left\{\mathbf{1}_{\{\tau^{Ins}<T\}}\cdot\mathcal{A}_{0,\tau^{Ins}}+\mathbf{1}_{\{
\tau^{Ins}\geq T\}}\cdot\mathcal{A}_{0,T}\right\}.
\end{align}
In contrast to the previous section on the discrete tenor structure,
the optimal stopping time problem (\ref{optimalstoppingtime2}) can
now only be stopped at the Cox random times $\{T_{n}\}_{n\geq 1
}\backslash T_{*}$. Hence, knowing only the Brownian filtration
$\{\mathcal{F}_t\}$ is certainly not enough to decide when to stop;
one has to know the additional filtration $\{\mathcal{H}_t\}$ from
the Cox process in order to determine when to stop. Similar to the
case of the discrete tenor structure, we can show that the solution
to this optimal stopping time problem is given by the dynamic
programming equation (\ref{DynamicProgrammingEquationMaturity}).

\begin{theorem}\label{theoremoptimalstoppingtime2}
The value of the optimal stopping time problem
(\ref{optimalstoppingtime2}) is given by the value function $U(0,x)$
in the dynamic programming equation
(\ref{DynamicProgrammingEquationMaturity}). The optimal stopping
time is given by the earliest maturity date such that the firm
fundamental falls below the debt run barrier in Proposition
\ref{TheoremForThreshold}:
$$\tau^{Run}=\inf\{T_n: V_{T_n}\leq D_{T_n}^{Run}\}\wedge T.$$
\end{theorem}

\begin{proof} See Appendix \ref{app D}.
\end{proof}

\subsection{Another Look at Default Mechanism}

{By deriving the debt run barrier and illiquidity barrier from the
DPP, together with the exogenous insolvency barrier, we obtain the
default mechanism in both Theorems \ref{Theorem1} and
\ref{Theorem2}. In this section, from the optimal stopping
representation of debt runs in both Theorems
\ref{theoremoptimalstoppingtime1} and
\ref{theoremoptimalstoppingtime2}, we can interpret the default
mechanism from the optimal stochastic control viewpoint as follows.
The representative creditor will choose an optimal rollover date to
withdraw her funding, i.e. to run on the firm. For the case of
discrete tenor structure, the creditor will choose an optimal
stopping time from a sequence of deterministic times. For the case
of staggered tenor structure, the creditor will choose an optimal
stopping time from a sequence of Cox arrival times. At each rollover
date, she can only make her decision if the firm is solvent up to
that date, and if the firm survives the debt run by other creditors
(based on her belief $\xi$). What we have shown is that the DPP used
to derive the debt run barrier strategy corresponds to a
non-standard optimal stopping time problem. Therein, the bankruptcy
time due to debt runs $T^*$ is based on the creditor's belief $\xi$,
so it is not necessarily the real bankruptcy time, and the creditor
can choose to run either before $T^*$ or after $T^*$. The creditor
will then decide her debt run barrier $D^{Run}$ based on her belief
$\xi$, or equivalently $T^*$. Since we consider the decision problem
of a representative short-term creditor, all of the short-term
creditors will run on the firm if $V_{T_n}\leq D^{Run}_{T_n}$.}

\section{Discussion and Conclusion}\label{SecConclusion}

In this paper, we provide a rigorous formulation for a class of
structural credit risk models that take into account not only
insolvency risk but also illiquidity risk due to possible debt runs.
We show that there exists a unique threshold strategy, i.e., a debt
run barrier for short-term creditors to decide when to withdraw
their funding. This allows us to decompose the total credit risk
into an illiquidity component based on the endogenous debt run
barrier and an insolvency component based on the exogenous
insolvency barrier.


%

The default mechanism in dynamic debt run models is mainly triggered
by creditors' runs as shown in \cite{MorrisShin}, \cite{HeXiongRun},
and \cite{LiangLutkebohmertXiao}. This is different from traditional
structural credit risk models where the default mechanism is usually
triggered by equity holders as they either exogenously set a default
barrier or endogenously determine an optimal default barrier.
\cite{Cheng} consider decision problems of both creditors
and equity holders in the dynamic debt run setting. In this paper,
we consider that the equity holders exogenously set the insolvency
barrier, while the creditors endogenously determine the debt run
barrier and the illiquidity barrier. On the other hand, most of
dynamic debt run models are based on the DPP, but up to now the
corresponding optimal stochastic control problem for the DPP in
dynamic debt run models has not been specified. In this paper, we
prove that the DPP is in fact derived from a non-standard optimal
stopping time problem with control constraints and we explicitly
state the associated optimal control problem. This may help us
better understand the default mechanism of debt runs.

In dynamic debt run models, one crucial assumption is the maturity
structure of short-term debt. Both \cite{HeXiongRun} and
\cite{Cheng} utilize the Poisson random maturity assumption to
capture the staggered tenor structure, whereas
\cite{LiangLutkebohmertXiao} assume a sequence of deterministic
rollover dates generalizing the two-period model of
\cite{MorrisShin}. In this paper, we consider both discrete and
staggered tenor structures. Moreover, we show that the two tenor
structures converge to each other when the rollover frequency goes
to infinity.

Finally, the representative short-term creditor's belief about other
creditors' current and future rollover decisions also characterizes
a dynamic debt run model. In \cite{MorrisShin} and
\cite{LiangLutkebohmertXiao} such a belief is modeled by a uniformly
distributed random variable. In this paper, we generalize this
assumption by modeling such a belief as a general random variable.
Furthermore, the impact of the creditors' future rollover decisions
are included by considering the dynamic programming equation of the
representative short-term creditor, which is in the same spirit as
\cite{HeXiongRun}. Notwithstanding, our model only takes account of
other creditors' rollover decisions on a representative creditor,
but not vice vera, by assuming her belief exogenously, because in
practice such a belief may depend on various factors that are not
present in the model such as monetary policy and the states of the
economy. Hence, in this sense our debt run model is an exogenous
model rather than an equilibrium model. The corresponding
equilibrium model could be more challenging, and is left for the
future research.

\section*{Acknowledgements}
\small{We thank the Editor-in-Chief, Ulrich Horst, a Co-Editor, an
Associate Editor, and a Referee for their valuable comments and
suggestions. The article was previously circulated under the title
\emph{A Continuous Time Structural Model for Insolvency, Recovery,
and Rollover Risks}. This work was supported by the Oxford-Man
Institute of Quantitative Finance, University of Oxford, and by the
Excellence Initiative through the project ``Pricing of Risk in
Incomplete Markets'' within the Institutional Strategy of the
University of Freiburg. The financial support is gratefully
acknowledged by the first and the second authors. Several helpful
comments and suggestions from Lishang Jiang, Yajun Xiao, and Qianzi
Zeng are very much appreciated. We also thank the participants at
the {\it Conference on Liquidity and Credit Risk} in Freiburg 2012,
the {\it INFORMS International Meeting} in Beijing 2012, the {\it
4th Berlin Workshop on Mathematical Finance for Young Researchers}
in Berlin 2012, the {\it 2013 International Conference on Financial
Engineering}, in Suzhou 2013, the {\it 6th Financial Risks
International Forum on Liquidity Risk} in Paris 2013, the {\it IMA
Conference on Mathematics in Finance} in Edinburgh 2013, the {\it
30th French Finance Association Conference} in Lyon 2013, and the
{\it European Financial Management Association 2013 Annual Meetings}
in Reading 2013, as well as seminar participants at the University
of Oxford, Imperial College, University of Texas at Austin, and
Tongji University for several insightful remarks.}


\appendix

\section{Appendix for Proofs}
\subsection{Proof of Proposition \ref{prop Green's representation}}\label{app A}

The proof is essentially the same as Lemma 3.2 and 3.3 in
\cite{LiangJiang}, so we only sketch it.

First, note that under the new coordinate $\bar x=x/(\beta l_t)$,
the PDEs (\ref{PenalizedEqu}) and (\ref{PenalizedEqu2}) become
$\mathcal{L}^vW_n=0$ on a regular domain
$[T_{n},T_{n+1}]\times[1,\infty)$. The Green's function
$\mathbb{G}(t,\bar{x};T_{n+1},\xi)$ for the operator $\mathcal{L}^v$
on $[T_{n},T_{n+1}]\times[1,\infty)$ is the solution to the
following PDE problem
\begin{equation}\label{Green}
\left\{
\begin{array}{l}
\mathcal{L}^v\mathbb{G}(t,\bar{x};T_{n+1},\xi)=0\\
\mathbb{G}|_{\bar{x}=1}=0\\
\mathbb{G}|_{t=T_{n+1}}=\delta(\bar{x}-\xi).%
\end{array}%
\right.
\end{equation}%
By making the transformation $y=\log(\bar{x}/\xi)$,
$\tau=T_{n+1}-t$, and
\begin{align*}
&\mathbb{G}(\tau,y;T_{n+1},\xi)\\
=&\ \exp\left\{\left[r_S-r-\frac{1}{2\sigma^2}
\left(r_V-r_L-\frac{\sigma^2}{2}\right)^2\right]\tau-\frac{1}{\sigma^2}\left(r_V-r_L-\frac{\sigma^2}{2}\right)y\right\}
\mathbb{H}(\tau,y;T_{n+1},\xi),
\end{align*}
it is easy to verify that $\mathbb{H}(\tau,y;T_{n+1},\xi)$ satisfies
a heat equation on the half plane. Its solution can be easily
obtained by the standard image method.

Next, given the Green's function $\mathbb{G}(\tau,y;T_{n+1},\xi)$,
we derive the solution to $\mathcal{L}^vW_n=0$ on the domain
$[T_{n},T_{n+1}]\times[1,\infty)$ with the boundary and terminal
data $P_n$ and $Q_n$ by applying integration by parts. Consider the
adjoint problem of (\ref{Green}) on $[t,T_{n+1}]\times[1,\infty)$
\begin{equation}\label{GreenAdjoint}
\left\{
\begin{array}{l}
\mathcal{\hat{L}}^v\mathbb{\hat{G}}(\eta,\xi;t,\bar{x})=0\\
\mathbb{\hat{G}}|_{\xi=1}=0\\
\mathbb{\hat{G}}|_{\eta=t}=\delta(\xi-\bar{x}),%
\end{array}%
\right.
\end{equation}%
where $\mathcal{\hat{L}}^v$ is the adjoint operator of
$\mathcal{L}^v$
$$\mathcal{\hat{L}}^v=-\frac{\partial}{\partial\eta}+\frac12\sigma^2\frac{\partial^2}{\partial\xi^2}\xi^2
-(r_V-r_L)\frac{\partial}{\partial\xi}\xi+(r_S-r).$$ Since
$W_n(\eta,\xi)$ satisfies $\mathcal{L}^vW_n=0$ and
$\hat{G}(\eta,\xi;t,\bar{x})$ satisfies the adjoint equation
$\mathcal{\hat{L}}^v\hat{G}(\eta,\xi;t,\bar{x})=0$, applying
integration by parts to the integral
$$
\int_{1}^{\infty}d\xi\int_{t+\epsilon}^{T_{n+1}-\epsilon}
[\hat{G}(\eta,\xi;t,\bar{x})\mathcal{L}^vW_n(\eta,\xi)
-W_n(\eta,\xi)\mathcal{\hat{L}}^v\hat{G}(\eta,\xi;t,\bar{x})]d\eta,
$$
and using the boundary and terminal data $P_n$ and $Q_n$ for
$W_n(\eta,\xi)$ will give us the Green's representation formula
(\ref{GreenFuntion}).\\

\subsection{Proof of Proposition \ref{Feynman-Kac3}}\label{app E}

We have the following property for the first arrival time (i.e., the
first short-term debt maturity) $T_1$, the proof of which can be
found for example in \cite{BieleckiRutkowski}.

\begin{lemma}\label{basiclemma}
The process $\Gamma$ defined by $\Gamma_t=\int_0^tg(X_s)ds$ for
$t\geq 0$ is an $\mathcal{F}_t$-hazard process associated with
$T_1$, that is,
$$\Gamma_t=-\log \mathbf{Q}(T_1>t|\mathcal{F}_t)=-\log\mathbf{Q}(T_1>t|\mathcal{F}_{\infty}).$$
Moreover, for any $\mathcal{F}_{t}$-adapted process $Y_t$ and
$\mathcal{F}_t$-stopping time $\tau$, on the event $\{T_1>t\}$,
\begin{equation}\label{basicformula1}
\mathbf{E}[\mathbf{1}_{\{T_1\geq
\tau\}}Y_{\tau}\left|\right.\mathcal{F}_t]
=\mathbf{E}\left[Y_{\tau}e^{-\int_t^{\tau}g(X_s)ds}\left|\right.\mathcal{F}_t\right].
\end{equation}
\begin{equation}\label{basicformula2}
\mathbf{E}\left[\mathbf{1}_{\{t<T_1<{\tau}\}}Y_{T_1}\left|\right.\mathcal{F}_t\right]
=\mathbf{E}\left[\int_t^{\tau}Y_{s}e^{-\int_t^sg(X_u)du}g(X_s)ds\left|\right.\mathcal{F}_t\right].
\end{equation}
\end{lemma}

In the following, we employ the distribution of $T_1$ given by Lemma
\ref{basiclemma} to calculate
(\ref{DynamicProgrammingEquationMaturity}). For the first and the
third terms, by using (\ref{basicformula2}), we obtain
\begin{align*}
&\ \mathbf{E}_t^{x}\left\{\mathbf{1}_{\{t<T_1<\tau^{Ins}\wedge T\}}
e^{(r_S-r)(T_1-t)}\left[\theta(X_{T_{1}})\max\left\{1,U(T_{1},X_{T_{1}})\right\}+(1-\theta (X_{T_1}))R_{T_1}\right]\right\}\\
=&\ \mathbf{E}_t^{x}\left\{\int_t^{\tau^{Ins}\wedge
T}e^{\int_t^s(r_S-r-g(X_u))du}g(X_s)\right.\\
&\times\left.\left[\theta(X_s)\max\{1,U(s,X_s)\}+(1-\theta(X_s))
\alpha X_{s}/(1+l_s)\right]ds\right\}.
\end{align*}
For the second term, based on (\ref{basicformula1}), we obtain
\begin{align*}
&\ \mathbf{E}_t^x\left\{\mathbf{1}_{\{ T_1\geq\tau^{Ins}, t\leq
\tau^{Ins}<T\}}e^{(r_S-r)(\tau^{Ins}-t)}R_{\tau^{Ins}}\right\}\\
=&\ \mathbf{E}_t^x\left\{\mathbf{1}_{\{t\leq
\tau^{Ins}<T\}}e^{\int_t^{\tau^{Ins}}(r_S-r-g(X_u))du}\alpha \beta
l_{\tau^{Ins}}/(1+l_{\tau^{Ins}})\right\}.
\end{align*}
For the last term, by employing (\ref{basicformula1}) again, we
obtain
\begin{align*}
&\ \mathbf{E}_t^x\left\{\mathbf{1}_{\{ T_1\geq T,
\tau^{Ins}\geq T\}}e^{(r_S-r)(T-t)}\min\left\{1, X_T/(1+l_T)\right\}\right\}\\
=&\ \mathbf{E}_t^x\left\{\mathbf{1}_{\{ \tau^{Ins}\geq
T\}}e^{\int_t^{T}(r_S-r-g(X_u))du}\min\left\{1,
X_T/(1+l_T)\right\}\right\}.
\end{align*}
By combing the above three equalities, we finally derive
\begin{align*}
U(t,x)=&\ \mathbf{E}_t^{x}\left\{\int_t^{\tau^{Ins}\wedge
T}e^{\int_t^s(r_S-r-g(X_u))du}g(X_s)
\right.\nonumber\\
&\times\left[\theta(X_s)\max\{1,U(s,X_s)\}+(1-\theta(X_s)) \alpha X_{s}/(1+l_s)\right]ds.\nonumber\\
&+\mathbf{1}_{\{t\leq
\tau^{Ins}<T\}}e^{\int_t^{\tau^{Ins}}(r_S-r-g(X_u))du}\alpha \beta l_{\tau^{Ins}}/(1+l_{\tau^{Ins}})\nonumber\\[+0.2cm]
&+\left.\mathbf{1}_{\{ \tau^{Ins}\geq
T\}}e^{\int_t^{T}(r_S-r-g(X_u))du}\min\left\{1,
X_T/(1+l_T)\right\}\right\}.
\end{align*}

Then similar to Proposition \ref{Feynman-Kac1}, the Feynman-Kac
formula gives us the PDE representation for the value function
$U(t,x)$ under the Cox maturity structure as provided in Proposition
\ref{Feynman-Kac3}.\\

\subsection{Proof of Theorem \ref{theoremoptimalstoppingtime1}}\label{app C}

For $n=0,1,\ldots,N$, we consider a sequence of optimal stopping
time problems
$$V(T_{n},x)=\sup_{\tau\in\{T_{n+1},T_{n+2},\ldots\}\backslash T_{*}}\mathbf{E}^x_{T_n}
\left\{\mathbf{1}_{\{T_n\leq\tau^{Ins}<T\}}\cdot\mathcal{A}_{T_{n},\tau^{Ins}}+\mathbf{1}_{\{
\tau^{Ins}\geq T\}}\cdot\mathcal{A}_{{T_n},T}\right\},$$ where
$\tau$ is an $\mathcal{F}_t$-stopping time taking value in
$\{T_{n+1},T_{n+2},\ldots\}\backslash T_{*}$. Then the value of the
optimal stopping time problem (\ref{optimalstoppingtime1}) is given
by $V(0,x)$, and we want to show that $V(0,x)=U(0,x)$.

Obviously we have $V(T_{N},x)=U(T_{N},x)$, since there is no
optimization problem involved in $V(T_{N},x)$ which is
$$
V(T_{N},x)=
\mathbf{E}_{T_{N}}^x\left\{\mathbf{1}_{\{T_{N}\leq\tau^{Ins}<T\}}e^{(r_S-r)(\tau^{Ins}-T_{N})}R_{\tau^{Ins}}+
\mathbf{1}_{\{\tau^{Ins}\geq T\}}e^{(r_S-r)(T-T_{N})}R_{T}\right\}.
$$
The idea is to introduce a sequence of auxiliary optimal stopping
time problems whose optimal stopping times are also permitted to
stop at the initial time $T_n$.
$$\hat{V}(T_{n},x)=\sup_{\tau\in\{T_{n},T_{n+1},\ldots\}\backslash T_{*}}\mathbf{E}^x_{T_n}
\left\{\mathbf{1}_{\{T_n\leq\tau^{Ins}<T\}}\cdot\mathcal{A}_{T_{n},\tau^{Ins}}+\mathbf{1}_{\{
\tau^{Ins}\geq T\}}\cdot\mathcal{A}_{{T_n},T}\right\}.$$ We have the
following relationship between $\hat{V}$ and $V$ (see \cite{Liang}):
\begin{equation}\label{relationship}
\hat{V}(T_n,x)=\theta(x)
\max\{1,V(T_{n},x)\}+(1-\theta(x))R_{T_{n}},\ \ \ \text{for}\
n=0,1,\ldots, N.
\end{equation}
For $n=0,1,\ldots,N-1$, by taking conditional expectation on
$\mathcal{F}_{T_{n+1}}$ in $V(T_{n},x)$, we obtain
\begin{align*}
V(T_{n},x)=&\ \sup_{\tau}\mathbf{E}^x_{T_n}
\left\{\mathbf{E}\left[\underbrace{\mathbf{1}_{\{T_n\leq\tau^{Ins}<T\}}\cdot\mathcal{A}_{T_{n},\tau^{Ins}}+\mathbf{1}_{\{
\tau^{Ins}\geq T\}}\cdot\mathcal{A}_{{T_n},T}}_{\equiv I}|\mathcal{F}_{T_{n+1}}\right]\right\}\\[+0.2cm]
=&\ \sup_{\tau}\mathbf{E}^x_{T_n}
\left\{\mathbf{E}\left[\left(\mathbf{1}_{\{T_n\leq\tau^{Ins}<T_{n+1}\}}+\mathbf{1}_{\{\tau^{Ins}\geq T_{n+1}\}}\right)\times I|\mathcal{F}_{T_{n+1}}\right]\right\}\\[+0.2cm]
=&\ \sup_{\tau}\mathbf{E}^x_{T_n}
\left\{\mathbf{1}_{\{T_n\leq\tau^{Ins}<T_{n+1}\}}e^{(r_S-r)(\tau^{Ins}-T_n)}R_{\tau^{Ins}}\right.\\
&+\left.\mathbf{1}_{\{\tau^{Ins}\geq
T_{n+1}\}}e^{(r_S-r)(T_{n+1}-T_{n})}
\mathbf{E}\left[\mathbf{1}_{\{T_{n+1}\leq\tau^{Ins}<T\}}\cdot\mathcal{A}_{T_{n+1},\tau^{Ins}}+\mathbf{1}_{\{
\tau^{Ins}\geq T\}}\cdot\mathcal{A}_{{T_{n+1}},T}|\mathcal{F}_{T_{n+1}}\right]\right\}\\[+0.2cm]
=&\ \sup_{\tau}\mathbf{E}^x_{T_n}
\left\{\mathbf{1}_{\{T_n\leq\tau^{Ins}<T_{n+1}\}}e^{(r_S-r)(\tau^{Ins}-T_n)}R_{\tau^{Ins}}\right.\\
&+\left.\mathbf{1}_{\{\tau^{Ins}\geq
T_{n+1}\}}e^{(r_S-r)(T_{n+1}-T_{n})}
\mathbf{E}_{T_{n+1}}^{X_{T_{n+1}}}\left[\mathbf{1}_{\{T_{n+1}\leq\tau^{Ins}<T\}}\cdot\mathcal{A}_{T_{n+1},\tau^{Ins}}+\mathbf{1}_{\{
\tau^{Ins}\geq T\}}\cdot\mathcal{A}_{{T_{n+1}},T}\right]\right\},
\end{align*}
where we used the Markovian property for $X$ in the last equality.
Note that the first term in the bracket does not involve the
stopping time $\tau$, so the supremum over $\tau$ only takes action
on the second term and $V(T_{n},x)$ is equal to
\begin{align*}
\ \mathbf{E}^x_{T_n}
&\left\{\mathbf{1}_{\{T_n\leq\tau^{Ins}<T_{n+1}\}}e^{(r_S-r)(\tau^{Ins}-T_n)}R_{\tau^{Ins}}+\mathbf{1}_{\{\tau^{Ins}\geq T_{n+1}\}}e^{(r_S-r)(T_{n+1}-T_{n})}\right.\\
&\times\left. \sup_{\tau\in\{T_{n+1},T_{n+2},\ldots\}\backslash
T_{*}}\mathbf{E}_{T_{n+1}}^{X_{T_{n+1}}}\left[\mathbf{1}_{\{T_{n+1}\leq\tau^{Ins}<T\}}\cdot\mathcal{A}_{T_{n+1},\tau^{Ins}}+\mathbf{1}_{\{
\tau^{Ins}\geq T\}}\cdot\mathcal{A}_{{T_{n+1}},T}\right]\right\},
\end{align*}
which, according to the definition of $\hat{V}$, is
$$\mathbf{E}^x_{T_n}
\left\{\mathbf{1}_{\{T_n\leq\tau^{Ins}<T_{n+1}\}}e^{(r_S-r)(\tau^{Ins}-T_n)}R_{\tau^{Ins}}+\mathbf{1}_{\{\tau^{Ins}\geq
T_{n+1}\}}e^{(r_S-r)(T_{n+1}-T_{n})}\hat{V}(T_{n+1},X_{T_{n+1}})\right\}.$$
By the relationship (\ref{relationship}), we obtain the recursive
formulation for $V(T_{n},x)$:
\begin{align*}
&\ V(T_{n},x)\\
=&\ \mathbf{E}_{T_n}^{x}\left\{\mathbf{1}_{\{T_n\leq
\tau^{Ins}<T_{n+1}\}}
e^{(r_S-r)(\tau^{Ins}-T_n)}R_{\tau^{Ins}}\right.\\
&+ \left.\mathbf{1}_{\{ \tau^{Ins}\geq
T_{n+1}\}}e^{(r_S-r)(T_{n+1}-t)}
\left[\theta(X_{T_{n+1}})\max\{1,V(T_{n+1},X_{T_{n+1}})\}+
(1-\theta(X_{T_{n+1}}))R_{T_{n+1}}\right]\right\}.
\end{align*}
We recognize that the above equation is just the dynamic programming
equation for $U(t,x)$ in (\ref{DynamicProgrammingEquation}). Since
we have already proved $V(T_{N},x)=U(T_{N},x)$, by proceeding
backwards we obtain
$V(0,x)=U(0,x)$.\\

\subsection{Proof of Theorem \ref{theoremoptimalstoppingtime2}} \label{app D}

The proof is essentially the same as the proof for Theorem
\ref{theoremoptimalstoppingtime1}. For any $t\geq 0$, by letting $X$
start from $X_t=x$ and $\{T_n\}_{n\geq 0}$ start from $T_0=t$, we
consider a family of optimal stopping problems
$$V(t,x)=\sup_{\tau\in\{T_{n}\}_{n\geq 1}\backslash T_{*}}\mathbf{E}^x_t
\left\{\mathbf{1}_{\{t\leq
\tau^{Ins}<T\}}\cdot\mathcal{A}_{t,\tau^{Ins}}+\mathbf{1}_{\{
\tau^{Ins}\geq T\}}\cdot\mathcal{A}_{t,T}\right\},$$ where $\tau$ is
a $\mathcal{G}_t$-stopping time taking value in $\{T_{n}\}_{n\geq
1}\backslash T_{*}$. Therefore, $\tau$ is not allowed to stop at the
starting time $t$. The value of the optimal stopping time problem
(\ref{optimalstoppingtime2}) is given by $V(0,x)$, and we want to
prove that $V(0,x)=U(0,x)$.\\
Similarly to the case of the discrete tenor structure, we introduce
a family of auxiliary optimal stopping time problems where the
optimal stopping times are also allowed to stop at the starting time
$t$
$$\hat{V}(t,x)=\sup_{\tau\in\{T_{n}\}_{n\geq 0}\backslash T_{*}}\mathbf{E}^x_t
\left\{\mathbf{1}_{\{t\leq
\tau^{Ins}<T\}}\cdot\mathcal{A}_{t,\tau^{Ins}}+\mathbf{1}_{\{
\tau^{Ins}\geq T\}}\cdot\mathcal{A}_{t,T}\right\}.$$ We have the
following relationship between $\hat{V}$ and $V$ (see \cite{Liang}):
\begin{equation}\label{relationship2}
\hat{V}(t,x)=\theta(x)\max\{1,V(t,x)\}+(1-\theta(x))R_{t}\ \ \
\text{for}\ t\in[0,T).
\end{equation}
Taking expectations conditional on $X_{T_{1}}$ in $V(t,x)$ and using
the strong Markov property for $X$, we obtain for any $t\in[0,T)$
\begin{align*}
V(t,x)=&\
\sup_{\tau}\mathbf{E}_t^x\left\{\mathbf{E}\left[(\mathbf{1}_{\{t<T_1<\tau^{Ins}\}}+\mathbf{1}_{\{T_1\geq\tau^{Ins}\}})
\mathbf{1}_{\{t\leq\tau^{Ins}<T\}}\cdot\mathcal{A}_{t,\tau^{Ins}}\right.\right.\\
&+\left.\left.(\mathbf{1}_{\{t<T_1<T\}}+\mathbf{1}_{\{T_1\geq T\}})
\mathbf{1}_{\{\tau^{Ins}\geq
T\}}\cdot\mathcal{A}_{t,T}|X_{T_1}\right]\right\}\\
=&\sup_{\tau}\mathbf{E}_t^x\left\{\mathbf{E}\left[\mathbf{1}_{\{t\leq\tau^{Ins}<T\}}\mathbf{1}_{\{t<T_1<\tau^{Ins}\}}\mathcal{A}_{t,\tau^{Ins}}|X_{T_1}\right]
+\mathbf{1}_{\{t\leq\tau^{Ins}<T\}}\mathbf{1}_{\{T_1\geq\tau^{Ins}\}}e^{(r_S-r)(\tau^{Ins}-t)}R_{\tau^{Ins}}\right.\\
&+\left.\mathbf{E}\left[\mathbf{1}_{\{\tau^{Ins}\geq
T\}}\mathbf{1}_{\{t<T_1<T\}}\mathcal{A}_{t,T}|X_{T_1}\right]+\mathbf{1}_{\{\tau^{Ins}\geq
T\}}\mathbf{1}_{\{T_1\geq T\}}e^{(r_S-r)(T-t)}R_{T}\right\}\\
=&\sup_{\tau}\mathbf{E}_t^x\left\{\mathbf{1}_{\{t<T_1<\tau^{Ins}\wedge
T\}}e^{(r_S-r)(T_1-t)}\mathbf{E}\left[\mathbf{1}_{\{T_1\leq\tau^{Ins}<T\}}\mathcal{A}_{T_1,\tau^{Ins}}
+\mathbf{1}_{\{\tau^{Ins}\geq
T\}}\mathcal{A}_{T_1,T}|X_{T_1}\right]\right.\\
&+\left.\mathbf{1}_{\{t\leq\tau^{Ins}<T\}}\mathbf{1}_{\{T_1\geq\tau^{Ins}\}}e^{(r_S-r)(\tau^{Ins}-t)}R_{\tau^{Ins}}+
\mathbf{1}_{\{\tau^{Ins}\geq T\}}\mathbf{1}_{\{T_1\geq
T\}}e^{(r_S-r)(T-t)}R_{T}\right\}\\
=&\sup_{\tau}\mathbf{E}_t^x\left\{\mathbf{1}_{\{t<T_1<\tau^{Ins}\wedge
T\}}e^{(r_S-r)(T_1-t)}\mathbf{E}_{T_1}^{X_{T_1}}\left[\mathbf{1}_{\{T_1\leq\tau^{Ins}<T\}}\mathcal{A}_{T_1,\tau^{Ins}}
+\mathbf{1}_{\{\tau^{Ins}\geq
T\}}\mathcal{A}_{T_1,T}\right]\right.\\
&+\left.\mathbf{1}_{\{t\leq\tau^{Ins}<T\}}\mathbf{1}_{\{T_1\geq\tau^{Ins}\}}e^{(r_S-r)(\tau^{Ins}-t)}R_{\tau^{Ins}}+
\mathbf{1}_{\{\tau^{Ins}\geq T\}}\mathbf{1}_{\{T_1\geq
T\}}e^{(r_S-r)(T-t)}R_{T}\right\},
\end{align*}
which by the definition of $\hat{V}$ is equal to
\begin{align*}
\mathbf{E}_t^x&\left\{\mathbf{1}_{\{t<T_1<\tau^{Ins}\wedge
T\}}e^{(r_S-r)(T_1-t)}\hat{V}(T_1,X_{T_1})\right.\\
&+\left.\mathbf{1}_{\{t\leq\tau^{Ins}<T\}}\mathbf{1}_{\{T_1\geq\tau^{Ins}\}}e^{(r_S-r)(\tau^{Ins}-t)}R_{\tau^{Ins}}+
\mathbf{1}_{\{\tau^{Ins}\geq T\}}\mathbf{1}_{\{T_1\geq
T\}}e^{(r_S-r)(T-t)}R_{T}\right\}.
\end{align*}
The result then follows from the relationship
(\ref{relationship2}).\\

\section{Appendix for the Numerical Approximation of the Solution to PDE (\ref{PenalizedEquMaturity})}\label{app B}

We first transform PDE (\ref{PenalizedEquMaturity}) by defining $y =
\log (x/\beta l_t),$ $\tau = T - t$ and $u(\tau,y) = U(t,x)$. Then
PDE (\ref{PenalizedEquMaturity}) reduces to
\begin{equation}\label{pde1}
    \fs{\pa u}{\pa \tau} = \fs{1}{2}\sigma^2\fs{\pa^2 u}{\pa
y^2}+(r_V-r_L-\frac{1}{2}\sigma^2)\fs{\pa u}{\pa y}+ (r_s-r-\zeta)u
+ \eta\max\{1,u\}+\kappa,
\end{equation}
where
\begin{align*}
\zeta(\tau,y) &= g(x);\ \eta (\tau,y) = g(x)\theta(x);\
\kappa(\tau,y) =g(x)(1-\theta(x)) \alpha x/(1+l_t),
\end{align*}
with boundary and initial conditions
\begin{align*}
    u(\tau,0) &= \alpha \beta l_{T-\tau}/(1+l_{T-\tau}) = P(\tau) = P;\\
    u(0,y) &= \min\{1,e^y\beta l_T/(1+l_T)\} = \Phi(y) = \Phi.
\end{align*}

In the following, we derive the implicit finite difference equation
for PDE (\ref{pde1}). Let $\Delta\tau$ denote the step size between
two updates of the value function $u$ in the time dimension.
Similarly, $\Delta y$ denotes the step size between grid points in
the space dimension of the value function $u$. The relevant range of
two variables is taken to be
$$(\tau,y)\in[0,T]\times[0,\bar{y}],$$
where $\bar{y}$ is a large constant such that realization of $y$
outside the region $[0,\bar{y}]$ occurs with negligible probability.
At each grid point, we define
$$u_{j}^n=u(n\Delta\tau,j\Delta y);$$
and the implicit finite difference equation for $(\ref{pde1})$ is
\begin{align}\label{pde11}
    \fs{u^{n+1}_{j}-u^{n}_{j}}{\Delta\tau} =&\ \fs{1}{2}\sigma^2\fs{u^{n+1}_{j+1}-2u^{n+1}_j+u^{n+1}_{j-1}}{\Delta y^2}
    + (r_V-r_L-\frac{1}{2}\sigma^2)\fs{u^{n+1}_{j+1}-u^{n+1}_{j-1}}{2\Delta
    y}\nonumber\\
    &+(r_s-r-\zeta_j^{n+1})u^{n+1}_{j} +
    \eta_{j}^{n+1}\max\{1,u^{n+1}_j\}+{\kappa_j^{n}},
\end{align}
where $$\zeta^{n}_j=G(n\Delta\tau,j\Delta y);\
\eta^{n}_j=\eta(n\Delta \tau,j\Delta y);\
\kappa^{n}_j=F(n\Delta\tau,j\Delta y)$$ for $0\leq n\leq T/\Delta
\tau$ and $0\leq j\leq \bar{y}/\Delta  y$. The corresponding
boundary and initial conditions are
\begin{eqnarray*}
  u_{0} = P;\ u_{\bar{y}/\Delta y} = 0;\ u^0 = \Phi,
\end{eqnarray*}
where $P$ and $\Phi$, with abuse of notation, denote the vectors
containing the discrete values of the boundary and initial
conditions, respectively.

The implicit finite difference equation (\ref{pde11}) can be
rewritten as the following nonlinear algebraic equation:
\begin{equation}\label{newton}
    Au^{n+1}-(\eta^{n+1},\max\{1,u^{n+1}\}) = \overline{\kappa}^n,
\end{equation}
where $\overline{\kappa}^n = \fs{1}{\Delta \tau}u^n + \kappa^n -
[cP,0,0,\ldots,0]^{*}$, and $A$ is a tridiagonal matrix:
\begin{equation*}
    A=\left(
        \begin{array}{ccccc}
          a_1 & b & 0 & \ldots & 0 \\
          c & a_2 & b & \ldots & 0 \\
          0 & c & a_3 & \ldots & \ldots \\
          \ldots &\ldots& \ldots &\ldots& \ldots\\
          0 & 0 & \ldots & a_{\bar{y}/\Delta y-2} & b \\
          0 & 0 & \ldots & c & a_{\bar{y}/\Delta y-1} \\
        \end{array}
      \right)
\end{equation*}
with
\begin{align*}
a_j &= \fs{1}{\Delta\tau}+\fs{\sigma^2}{\Delta
y^2}-(r_S-r-\zeta^{n+1}_j);\\
b &= -\fs{1}{2}\fs{\sigma^2}{\Delta y^2}-\fs{1}{2\Delta
y}(r_V-r_L-\frac12\sigma^2),\\
c& = -\fs{1}{2}\fs{\sigma^2}{\Delta y^2}+\fs{1}{2\Delta
y}(r_V-r_L-\frac12\sigma^2).
\end{align*}

Finally, for $n=0,1,\ldots,T/\Delta\tau$, we use the standard Newton
method to solve the nonlinear algebraic equation (\ref{newton}) as
follows.
\begin{itemize}
  \item Set $v^{1} = u^n$;
  \item For $m=1,2,\ldots$, solve $v^{m+1}$ recursively by the corresponding linear
  equation for (\ref{newton})
$$Av^{m}-(\eta^{n+1},\max\{1,v^m\})-\bar{\kappa}^n+B^{n+1}_m(v^{m+1}-v^m)=0$$
until $\sup|v^{m+1}-v^m|<\epsilon$, where
  \begin{equation*}
    B^{n+1}_m = A - \eta^{n+1}\left(
              \begin{array}{ccccc}
                \mathbf{1}_{\{v_1^m>1\}} & 0 & \ldots & 0 & 0 \\
                0 & \mathbf{1}_{\{v_2^m>1\}} & 0 & \ldots & 0 \\
                0 & 0 & \ddots & 0 & 0 \\
                0 & 0 & 0 & 0 & \mathbf{1}_{\{v_{\bar{y}/\Delta y-1}^m>1\}} \\
              \end{array}
            \right)
  \end{equation*}
  \item Suppose the above loop runs $M$ times. Then set $u^{n+1}=v^M$.
\end{itemize}

\bibliographystyle{harvard}

\newpage
\setcounter{figure}{1}
\begin{center}
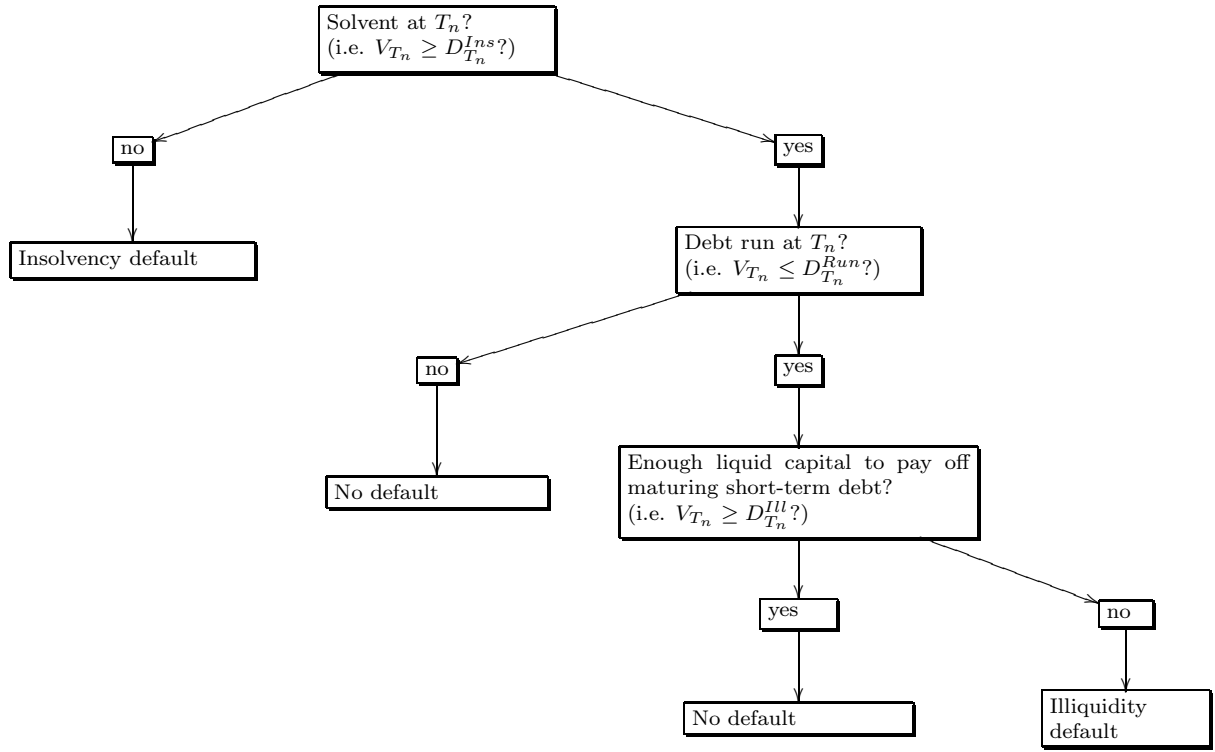
\begin{figure}[H]
\caption{\small\textbf{Scenarios at each rollover date
$T_n$}}\label{fig:flowchart} {\footnotesize
\[
\xymatrix{
                                             & *+[F-,]{\parbox{29mm}{Solvent at $T_n$?\\ (i.e. $V_{T_n}\geq D_{T_n}^{Ins}$?)}}\ar[dl]\ar[dr] \\
*+[F-,]{\hbox{no}}\ar[d] &                                                     & *+[F-,]{\hbox{yes}}\ar[d]\\
*+[F-,]{\parbox{30mm}{Insolvency default}} && *+[F-,]{\parbox{30mm}{Debt run at $T_n$?\\ (i.e. $V_{T_{n}}\leq D_{T_n}^{Run}$?)}}\ar[dl]\ar[d]\\
&  *+[F-,]{\hbox{no}}\ar[d]                       & *+[F-,]{\hbox{yes}}\ar[d]\\
&  *+[F-,]{\parbox{27mm}{No default}} & *+[F-,]{\parbox{45mm}{Enough liquid capital to pay off maturing short-term debt?\\ (i.e. $V_{T_n}\geq D_{T_n}^{Ill}$?)}}\ar[d]\ar[rd]\\
&
 &*+[F-,]{\parbox{8mm}{yes}}\ar[d] &*+[F-,]{\parbox{5mm}{no}}\ar[d] \\
&                                                               &*+[F-,]{\parbox{28mm}{No default}}&*+[F-,]{\parbox{20mm}{Illiquidity default}} \\
}
\]
}
\end{figure}
\end{center}

\begin{figure}[H]
\begin{center}
    \includegraphics[width=0.8\textwidth]{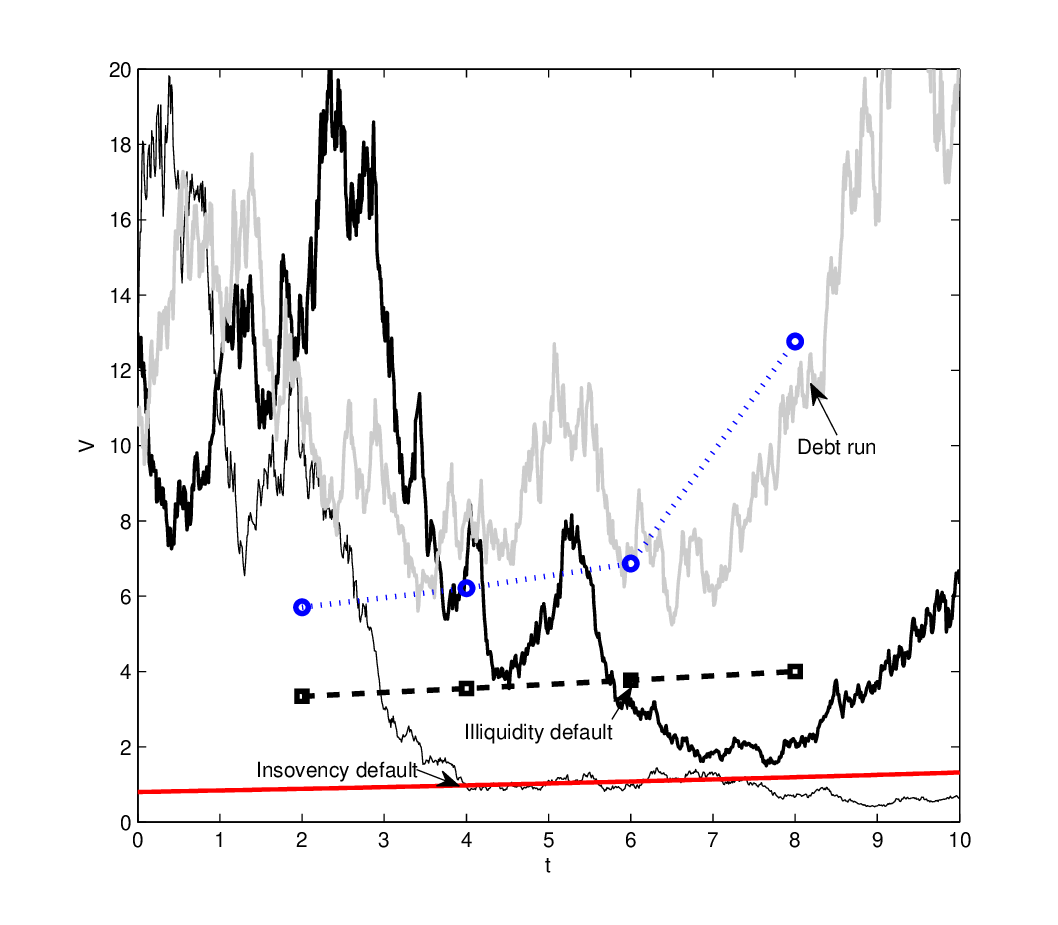}
    \caption{\label{fig: simulation discrete}\small Scenario simulation with discrete tenor structure\vspace{2mm}\newline
The figure shows three simulated asset value paths in the model with
a discrete tenor structure, where volatility $\sigma=0.4$, expected
return rate $r_V=-0.02$, market interest rate $r=0.01$, short-term
rate $r_S=0.03$, and long-term rate $r_L=0.05$. The initial values
of short- and long-term debt are set to $S_0=2$ and $L_0=2$,
respectively. The safety covenant parameter $\beta=0.4$, the
bankruptcy cost parameter $\alpha=0.6$, and the fire-sale rate is
set to $\psi=0.6$. The time horizon is $T=10$ years. In this
discrete tenor structure setting the number of rollover dates is set
to $N=4$. The dotted line describes the debt run barrier, the dashed
line the illiquidity barrier, and the solid line the insolvency
barrier.\vspace{2mm}}
\end{center}
\end{figure}
\begin{figure}[H]
\begin{center}
    \includegraphics[width=0.8\textwidth]{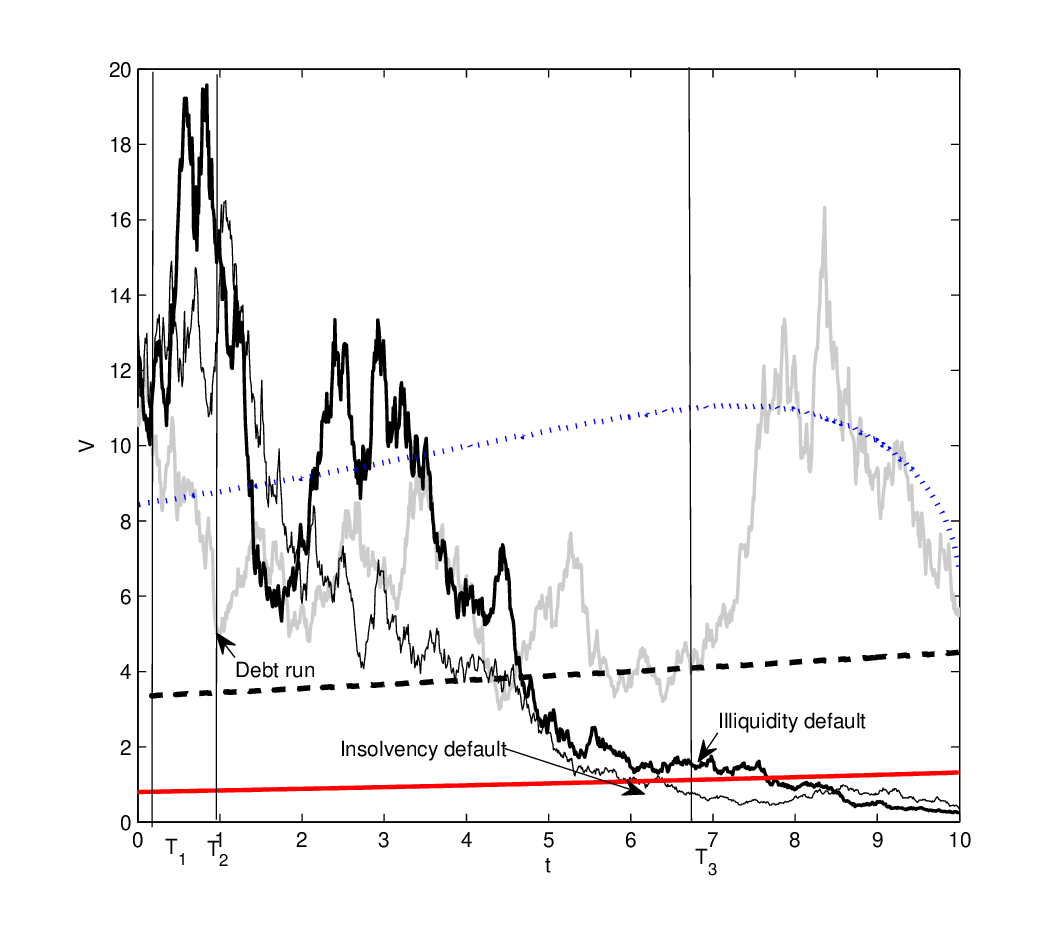}
    \caption{\label{fig: simulation continuous}\small Scenario simulation with a staggered tenor structure\vspace{2mm}\newline
The figure shows three simulated asset value paths in the model with
a staggered tenor structure where volatility $\sigma=0.4$, expected
return rate $r_V=-0.02$, market interest rate $r=0.01$, short-term
rate $r_S=0.03$, and long-term rate $r_L=0.05$. The initial values
of short- and long-term debt are set to $S_0=2$ and $L_0=2$,
respectively. The safety covenant parameter $\beta=0.4$, the
bankruptcy cost parameter $\alpha=0.6$, and the fire-sale rate is
set to $\psi=0.6$. The time horizon is $T=10$ years. In this
staggered tenor structure setting the intensity of the Cox process
is chosen to be $g(X_t)\equiv 0.4$. The dotted line describes the
debt run barrier, the dashed line the illiquidity barrier, and the
solid line the insolvency barrier. \vspace{2mm}}
\end{center}
\end{figure}

\begin{figure}
\begin{center}
    \includegraphics[width=1.0\textwidth]{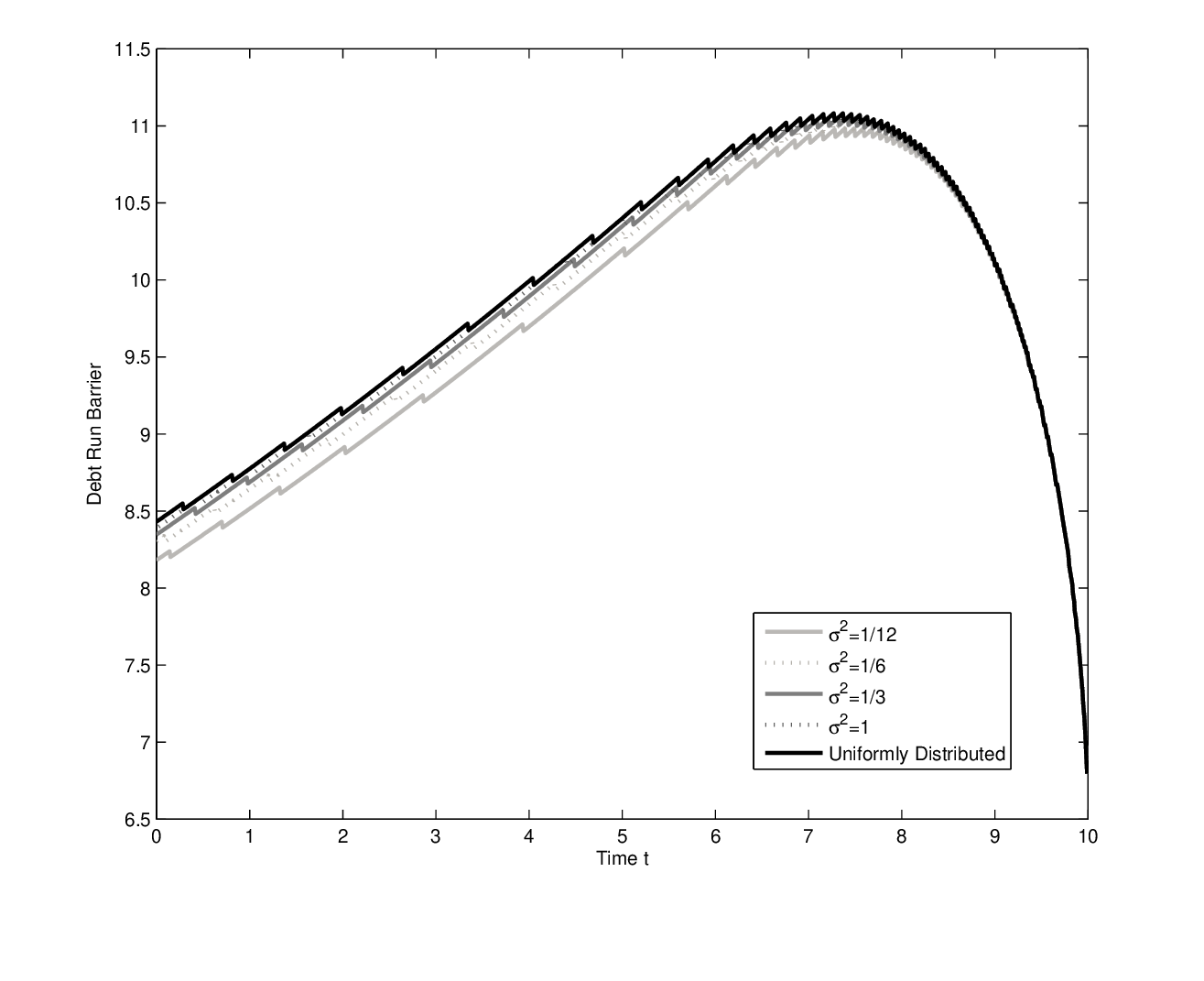}
\caption{\label{fig:distribution}\small Comparison of debt run
barriers with different assumptions on $\xi$\vspace{2mm}\newline The
figure shows debt run barriers with different assumptions on the
random variable $\xi$, the creditor's belief on the proportion of
short-term creditors not rolling over their funding at each rollover
date. The top line corresponds to the uniformly distributed $\xi$,
and from the second top to the bottom lines, they correspond to the
truncated normal distribution with mean $0.5$ and diminishing
variances. Other parameters are the same as in Figure \ref{fig:
simulation continuous}.}
\end{center}
\end{figure}


\begin{figure}[htbp]
 \centering
  \subfigure[For $\sigma = 0.2$ and $g(x)=0.2$]{
   \label{fig: PD11}
    \hspace{-3mm}
    \includegraphics[width=0.5\linewidth,height=0.35\textwidth]{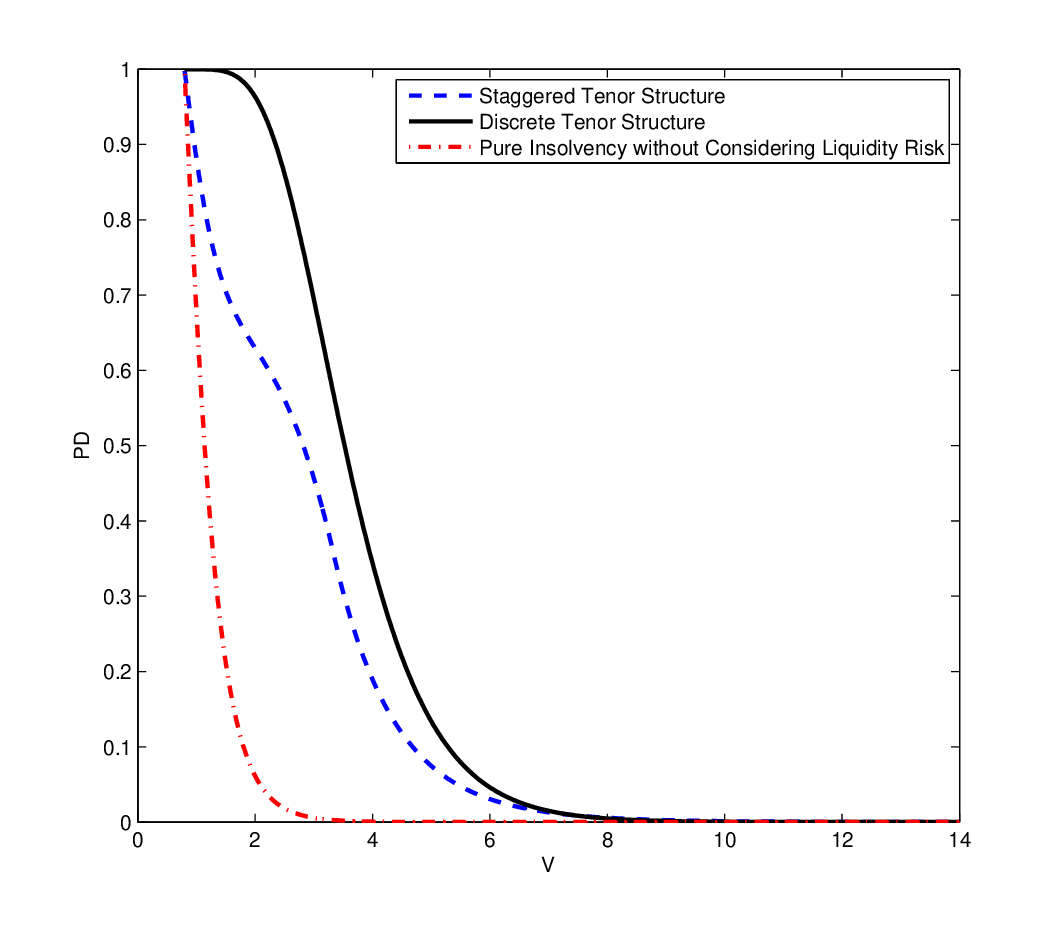}
     \hspace{-3mm}
  }
  \subfigure[For $\sigma = 0.2$ and $g(x)=0.4$]{
   \label{fig: PD5}
    \hspace{-3mm}
    \includegraphics[width=0.5\linewidth,height=0.35\textwidth]{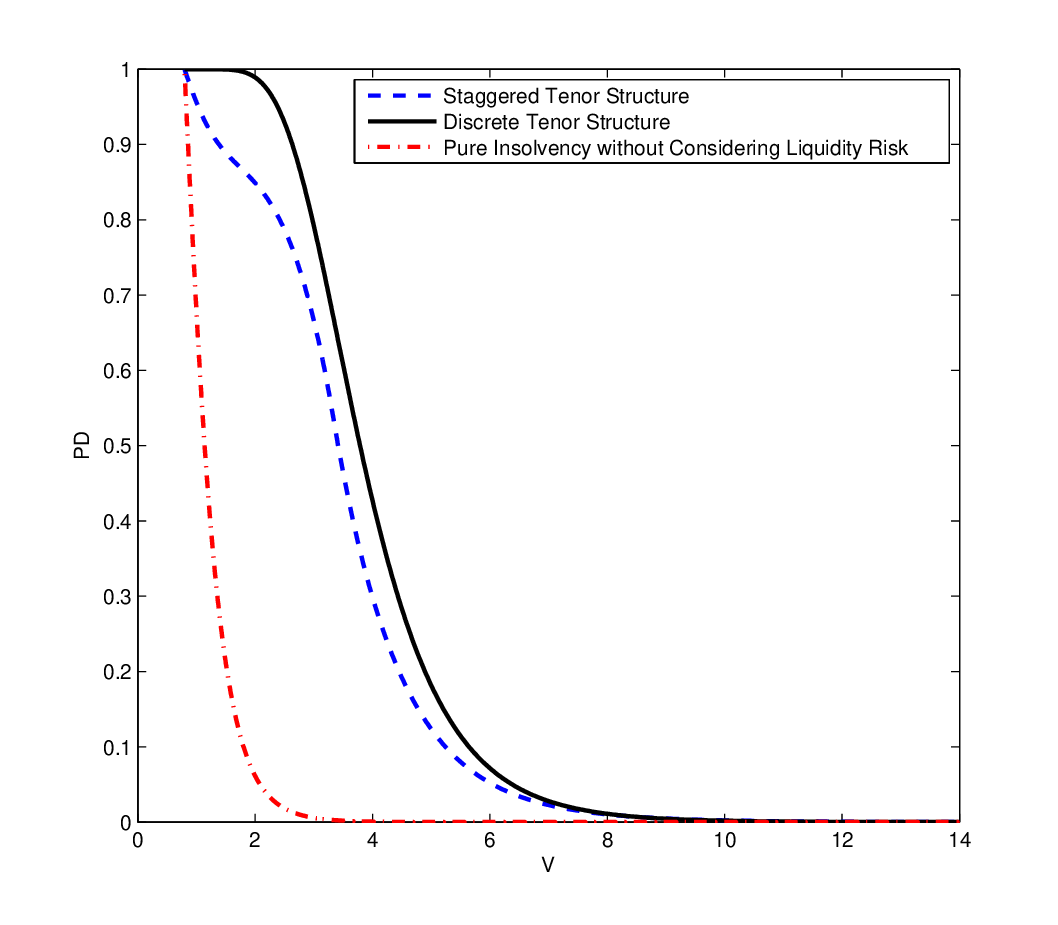}
     \hspace{-3mm}
  }
  \subfigure[For $\sigma = 0.4$ and $g(x)=0.2$]{
    \label{fig: PD7}
    \hspace{-3mm}
    \includegraphics[width=0.5\linewidth,height=0.35\textwidth]{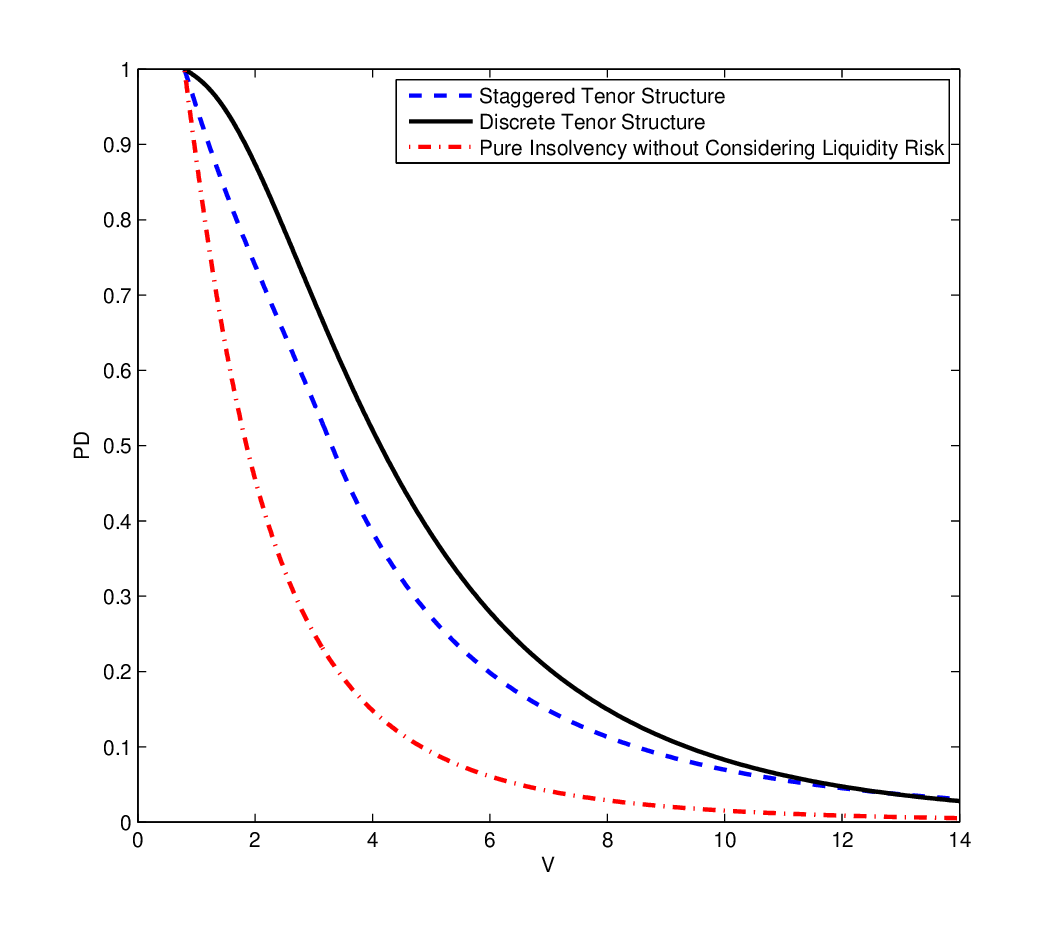}
    \hspace{-3mm}
    }
    \subfigure[For $\sigma = 0.4$ and $g(x)=0.4$]{
   \label{fig: PD9}
    \hspace{-3mm}
    \includegraphics[width=0.5\linewidth,height=0.35\textwidth]{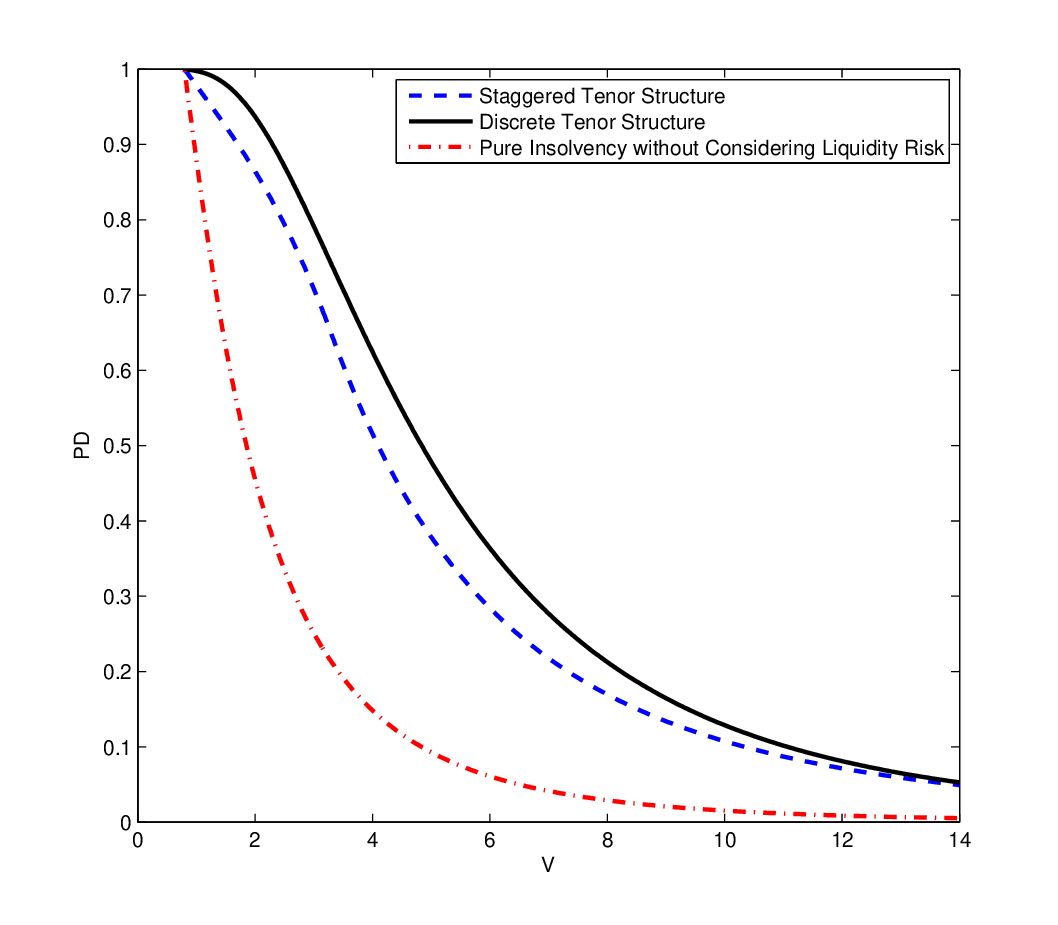}
     \hspace{-3mm}
  }
    \caption{\label{fig:PDs discrete and staggered}\small Default probabilities under the discrete and the staggered tenor structure
models\vspace{2mm}\newline The figure shows the default probabilies
under the discrete tenor structure and the staggered tenor
structure. The dotted line is the default probability without
including rollover risk. Other parameters are the same as in Figure
\ref{fig: simulation discrete} and Figure \ref{fig: simulation
continuous} apart from $T=5$ and $r_V = 0.07$ and with volatility
$\sigma$ and intensity $g(x)$ as specified below the
graphs.\vspace{2mm}}
\end{figure}

\begin{figure}[htbp]
 \centering
  \subfigure[For $\sigma = 0.2$ and $g(x)=0.4$]{
   \label{fig: PD6}
    \hspace{-3mm}
    \includegraphics[height=0.5\textwidth]{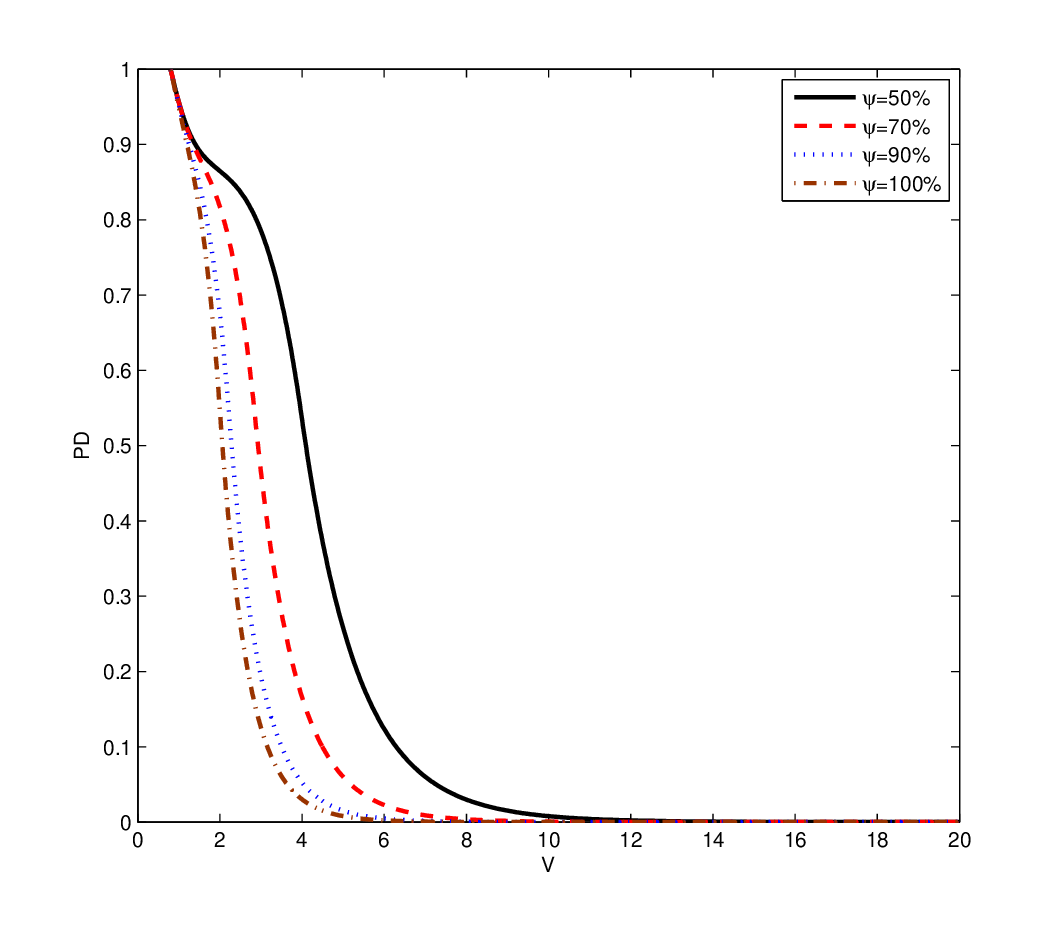}
     \hspace{-3mm}
  }
  \subfigure[For $\sigma = 0.4$ and $g(x)=0.2$]{
    \label{fig: PD8}
    \hspace{-3mm}
    \includegraphics[height=0.5\textwidth]{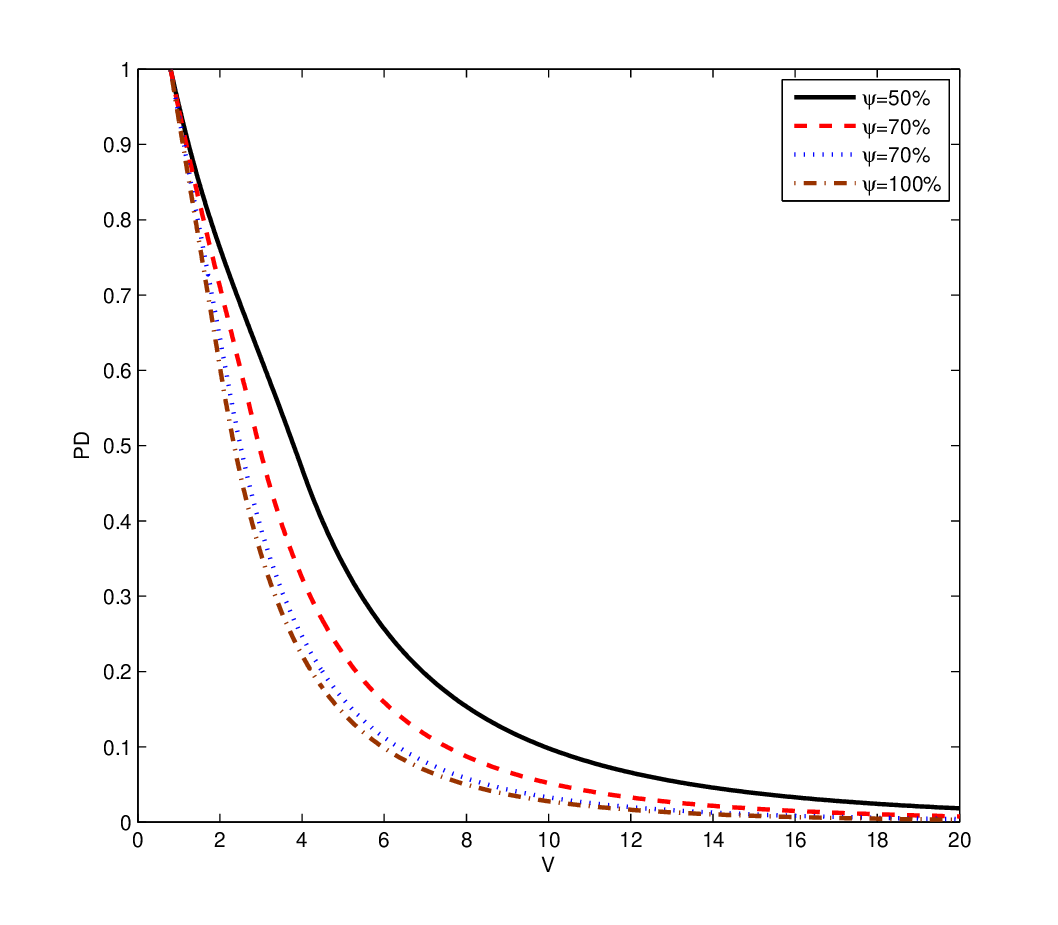}
    \hspace{-3mm}
    }
    \caption{\label{fig:PDs for different psi}\small Default probabilities for different fire-sale rates\vspace{2mm}
    \newline The figure shows the default probabilies for different fire-sale
    rates.
Other parameters are the same as in Figure \ref{fig: simulation
discrete} and Figure \ref{fig: simulation continuous} apart from
$T=5$ and $r_V = 0.07$ and with volatility $\sigma$ and intensity
$g(x)$ as specified below the graphs.\vspace{2mm}}
\end{figure}

\begin{figure}
\begin{center}
    \includegraphics[width=1.0\textwidth]{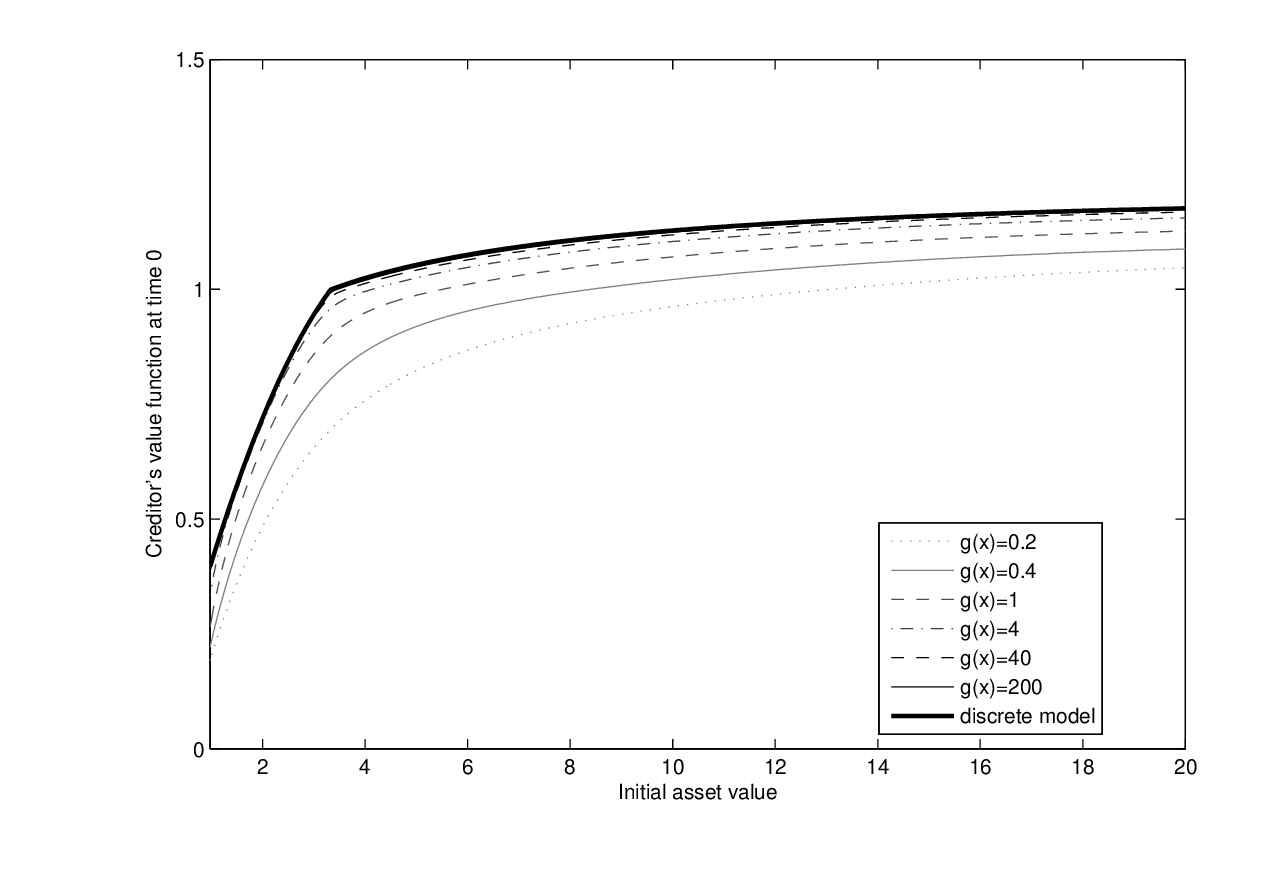}
\caption{\label{fig comparison variable g}\small Influence of the
intensity on creditor's value function\vspace{2mm}\newline The
figure shows the representative creditor's value function at time
$t=0$ for increasing initial asset value $V_0$ in the discrete tenor
structure model with $N=1000$ rollover dates and for the staggered
tenor structure model for different intensities $g$ of the Cox
process. Other parameters are the same as in Figure~\ref{fig:
simulation continuous}.}
\end{center}
\end{figure}

\begin{figure}
\begin{center}
    \includegraphics[width=1.0\textwidth]{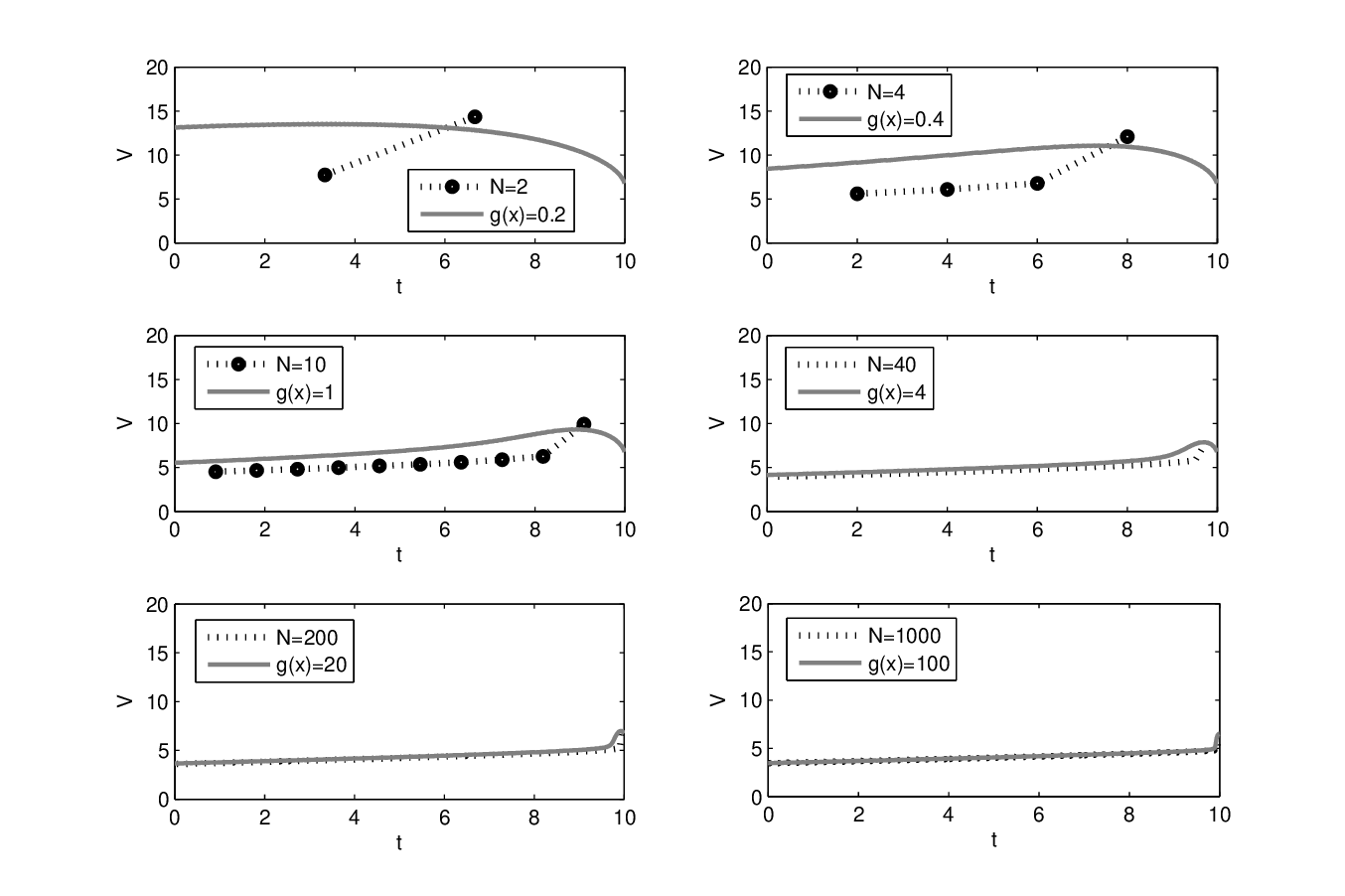}
\caption{\label{fig comparison variable g and N}\small Comparison
between the discrete and the staggered tenor structure
models\vspace{2mm}\newline The figure shows the debt run barrier
depending on time $t$ for different rollover frequencies $N$ in the
discrete tenor structure model and for different intensities $g(x)$
of the Cox process in the staggered tenor structure model. Other
parameters are the same as in Figure \ref{fig: simulation
continuous}.}
\end{center}
\end{figure}

\end{document}